\documentclass[preprint,11pt,authoryear]{elsarticle}
\usepackage{tikz}
\usepackage{xcolor}
\usepackage{amsfonts}
\usepackage{amsmath}
\usepackage{amsthm}
\usepackage{amssymb}
\usepackage{times}
\usepackage{subfig}
\usepackage{multirow}
\usepackage{setspace}
\usepackage{enumerate}
\usepackage{natbib}
\usepackage{color}
\usepackage{colortbl}
\usepackage{ wasysym }
\usepackage{verbatim}
\usetikzlibrary{shapes,arrows,calc,positioning}
 \usepackage{lscape}
\usepackage{array}
\usepackage{setspace}

\newcolumntype{x}[1]{%
>{\centering\hspace{0pt}}p{#1}}%

\def\eqd{\buildrel {\rm d} \over =}

\def\cw#1 { \overset{{P}}{\underset{#1}{\longrightarrow}} }
\def\Real{\mathbb{R}}
\def\P#1{{{P}}\left(#1\right)}

\def\E#1{{E}\left[#1\right]}

\def\Var#1{{\mathrm Var}\left(#1\right)}

\def \hix#1{ {\rm HIX}\left(#1\right) }
\def \six#1{ {\rm SIX}\left(#1\right) }

\def \rhix#1{ {{\rm RHIX}}\left(#1\right) }

\def \rcov#1#2 {{\rm cov}_{#1}\left( #2\right)}

\newtheorem{notation}{Notation}
\newtheorem{example}{Example}

\newtheorem{lemma}{Lemma}
\newtheorem{theorem}{Theorem}
\newtheorem{definition}{Definition}
\newtheorem{corollary}{Corollary}
\newtheorem{remark}{Remark}

\newtheorem{assumption}{Assumption}
\newtheorem*{toy*}{Toy Model}

\def\cov#1{{\rm  cov}\left[#1\right]}

\begin{document}
\begin{frontmatter}
\title{Negative Dependence Concept in Copulas and the Marginal Free Herd Behavior Index}

\author[EH]{Jae Youn Ahn\corref{cor2}}
\ead{jaeyahn@ewha.ac.kr}
\address[EH]{Department of Statistics, Ewha Womans University, 11-1 Daehyun-Dong, Seodaemun-Gu, Seoul 120-750, Korea.}

\cortext[cor2]{Corresponding Author}

\begin{abstract}
We provide a set of copulas that can be interpreted as having the negative extreme dependence.
This set of copulas is interesting because it coincides with countermonotonic copula
for a bivariate case, and more importantly, is shown to be minimal in concordance ordering in the sense that
no copula exists which is strictly smaller than the given copula outside the proposed copula set.
Admitting the absence of the minimum copula in multivariate dimensions greater than 2,
the study of the set of minimal copulas can be important in the investigation of various optimization problems.
To demonstrate the importance of the proposed copula set,
we provide the variance minimization problem of the aggregated sum with arbitrarily given uniform marginals.
As a financial/actuarial application of these copulas, we define
a new herd behavior index using weighted Spearman's rho, and determine the sharp lower bound of the index
using the proposed set of copulas.

\end{abstract}

\end{frontmatter}

\vfill

\pagebreak

\vfill

\pagebreak
\section{Introduction}\label{sec.1}
The study of the dependence structure between random variables via copula
is a classical problem in statistics and other applications.
The ease of application of copulas has led to their popularity in various areas such as finance, insurance, hydrology and
medical studies; see for example, \citet{Frees}, \citet{Genest} and \citet{Cui}.
This paper examines the mathematical property of copulas by focusing on their lower bound.

Every copula is bounded by Fr\'echet-Hoeffding lower and upper bounds. While Fr\'echet-Hoeffding upper bound corresponds to
the maximum copula, Fr\'echet-Hoeffding lower bound is generally not a copula.
Further, the minimum copula does not exist in general in high dimensions greater than $2$; see, for example, \citet{Kotz2} and \citet{Joe2}.

In the insurance and finance field, the maximum copula corresponds to the concept called comonotonicity \citep{Dhaene}.
In the respect of risk management, comonotonicity is an important concept, because it
can be used to
describe the perfect positive dependence between competing risks. Importantly it provides the solution
to various optimization (maximization) problems.
However, unlike the perfect positive dependence, mainly due to the absence of the minimum copula,
controversy has remained even in the definition of negative extreme dependence
In spite of these difficulties, the need for the concept of negative extreme dependence has remained in insurance and other
applications because it may lead to solutions for related optimization problems.
Many studies have investigated the negative extreme dependence in various contexts.
\citet{Dhaene4}, \citet{Cheung4} and \citet{Cheung6} defined the concept of mutual exclusivity which can be regarded as pairwise countermonotonic movements. On the other hand, \citep{Ruodu} proposed the concept of complete mixability, which can be used to minimize the variance
of the sum of random variables with given marginal distributions.
Many papers have recently been published in this field \citep{Ruodu2, Ruodu3, Ruodu4, Ruodu5}.
While the concepts of mutual exclusivity and complete mixability are both useful in various fields of optimization problems, since their concepts
both depend on the marginal distributions and are problem specific,
they may not provide the general concept of negative dependence.

\citet{Ahn7} proposed a set of negative dependence joint distributions,
which is named as $d$-countermonotonic copulas ($d$-CM).
The definition of $d$-CM is known to be the definition of copula only.
Furthermore, the set of $d$-CM copulas is minimal in terms of concordance ordering:
there is no copula which is strictly smaller in concordance ordering
than the given $d$-CM copula except $d$-CM copulas.
Admitting the absence of the minimum element in multivariate dimensions $d\ge 3$,
the set of minimal copulas can be important in optimization problems.
However, without understanding the further properties of $d$-CM copulas,
choosing the proper $d$-CM copulas for the given optimization problem can be difficult.
Furthermore, as specified in \citet{Ruodu3}, $d$-CM can be too general to be used for the negative extreme dependence
For example, any vector $(V, V, \cdots,V, 1-V)$ with $V$ being a uniform[0,1] random variable is $d$-CM, while it is close to a comonotonic random vector
except the last element.
Hence in this paper, to remove such an almost comonotonic case and emphasize the negative extreme dependence concept,
we consider only a special subset of $d$-CM copulas, 
which will be parameterized by the vector $\overrightarrow{w}\in \Real_+^d$,
where $\Real_+^d$ is a $d$-dimensional positive Euclidean space.
Such set of copulas will be named as $\overrightarrow{w}$-countermonotonic copulas ($\overrightarrow{w}$-CM).
Due to the minimality property of the set of $\overrightarrow{w}$-CM copulas, which is inherited from $d$-CM, we expect that the set of $\overrightarrow{w}$-CM copulas
might be also useful in various optimization problems.

However, before we discuss the usefulness of $\overrightarrow{w}$-CM copulas in optimization problems,
the existence of $\overrightarrow{w}$-CM copulas should be first investigated.
While the existence of $\overrightarrow{w}$-CM copulas with 
$$\overrightarrow{w}=(1,\cdots, 1)\in \Real_+^d$$
is well known in the literature,
see, for example, \citet{Ahn7}, existence of $\overrightarrow{w}$-CM copulas is not guaranteed for general $\overrightarrow{w}\in \Real_+^d$.
This paper provides the equivalence condition for the existence of $\overrightarrow{w}$-CM copulas. For the proof and construction of the copula,
we use a simple geometrical method to construct the copula.
A similar result obtained by using an algebraic method can be found in a recent working paper by \citet{Ruodu4}.

Since $\overrightarrow{w}$-CM is the property of the copula only,
the usefulness of $\overrightarrow{w}$-CM may be limited to some optimization problems which do not depend on marginal distributions.
\citet{Ruodu3} also note the possible limitedness of $\overrightarrow{w}$-CM (hence $\overrightarrow{d}$-CM) in solving optimization problems
by commenting that {\it any dependence concept which does not take into account marginal distributions may fail to
solve optimization problems which depend on marginal distributions}. Variance minimization of the aggregated sum with given marginal distributions,
which is formally stated in \eqref{var.eq} below, is one such example;
detailed literature can be found in \citet{Ruschendorf, Ruschendorf2, Ruodu, Ruodu3}.
As can be intuitively expected, and as will be shown in Section \ref{optim} below, it can be shown that no single copula
universally minimizes the variance of the aggregated sum with arbitrarily given marginals.
However, we will show that using a set of $\overrightarrow{w}$-CM copulas rather than a single copula
can minimize the variance of the aggregated sum for varying marginal distributions when restricted to the uniform distribution family.
While our result provides a general solution with no restriction on $\overrightarrow{w}\in\Real_+^d$,
a partial solution can be observed in \citet{Ruodu4} for some special cases of $\overrightarrow{w}\in\Real_+^d$
where they are mainly interested in so called {\it joint mixability} which aims for the constant aggregated sum.
More detailed results will be provided in Section \ref{optim}.

For a financial application of $\overrightarrow{w}$-CM, we provide a new definition of the herd behavior index.
Herd behaviors describe the comovement of members in a group.
Since herd behaviors in the stock markets are observed usually during financial crises \citep{Dhaene2, Ahnhix},
measuring the herd behavior can be important in managing financial risks.
Focusing on the fact that the perfect herd behavior can be modeled with the comonotonicity,
some herd behavior indices that measure the degree of comonotonicity via the concept of the (co)variance have been
proposed \citep{Dhaene2,Daniel14, Ahnhix}.
Measuring the herd behavior using such herd behavior indices can be important as it has been
shown to be an indicator of the market fear.
However, while the concept of comonotonicity is free of marginal distribution (and hence so is the herd behavior),
these herd behavior measures can depend on marginal distributions, as will be shown in Example \ref{ex..1} below. Alternatively,
we define the new herd behavior index based on a weighted average of bivariate Spearman's rho.
This new herd behavior index is not affected by the marginal distributions by definition and will be shown to preserve the concordance ordering.
We also show that the maximum and minimum of the new herd behavior are closely related with comonotonicity and $\overrightarrow{w}$-CM.

The rest of this paper is organized as follows.
We first summarize the study notations, and briefly explain basic copula theory and countermonotonicity theory in Section \ref{sec.2}.
The concept of $\overrightarrow{w}$-CM is introduced in Section \ref{sec.3}, and
the existence of $\overrightarrow{w}$-CM copula is demonstrated in Section \ref{sec.cond}.
Section \ref{optim} applies the concept of $\overrightarrow{w}$-CM to variance minimization problems.
The definition and minimization of the new herd behavior index are discussed in in Section \ref{hix.section}, which is followed by the conclusions.

\section{Notations and Preliminary Results}\label{sec.2}
\subsection{Conventions}
Let $d\ge 2$ be integers and $\Real^d$ denotes $d$ dimensional Euclidean space. Especially, let
$\Real_+^d$ be $d$ dimensional positive Euclidean space.
Further $[a,b]\times [a,b]\cdots\times [a,b] \subseteq \Real^d$ is denoted by $[a, b]^d$.
We use $\overrightarrow{\cdot}$ to denote $d$-variate vectors: especially,
lower case $$\overrightarrow {x}=(x_1,x_2,\cdots,x_d)$$ denotes constant vectors in $\Real^d$ and
upper case $$\overrightarrow {X}=(X_1,X_2,\cdots,X_d)$$ denotes $d$-variate random vectors.
More specifically
$$\overrightarrow{u}:=(u_1, \cdots, u_d)\quad\hbox{and}\quad \overrightarrow{w}:=(w_1, \cdots, w_d)$$
will be used to denote constant vectors in $[0,1]^d$ and $\Real_+^d$, respectively.
Finally, use $V$ to denote a uniform$[0,1]$ random variable.

Unless specified, we assume $\overrightarrow {X}$ be a $d$-dimensional random vector having
$H$ as its cumulative distribution function defined by
\[
H(\overrightarrow {x})=P(X_1 \leq x_1 , \cdots, X_d \leq x_d)
\quad \hbox{for} \quad\overrightarrow {x}\in \Real^d,
\]
and the marginal distribution of $X_i$ is $F_i (y):=P(X_i \leq y)$ for $i\in\{1, \cdots, d\}$ and $y\in\Real$.
Define $\mathcal{F}_d(F_1, \cdots, F_d)$ to be the Fr\'echet space of $d$-variate random vectors with marginal distribution $F_1, \cdots, F_d$.
Hence, $\overrightarrow{X}\in \mathcal{F}_d(F_1, \cdots, F_d)$.
Equivalently, we also denote $H\in \mathcal{F}_d(F_1, \cdots, F_d)$.
We use $\mathcal{F}_d$ to denote the special case of Fr\'echet space, where all marginal distributions are uniform$[0,1]$.

This paper assumes that marginals distributions are continuous.
According to Sklar (1959), given $H\in\mathcal{F}(F_1, \cdots, F_d)$, there exists a unique function $C:[0,1]^d\rightarrow [0,1]$ satisfying
\[
H({\overrightarrow{x}})=C(F_1 (x_1),\cdots,F_d(x_d)).
\]
The function $C$ is called a copula, which is also a distribution function on $[0,1]^d$.
Further information on copulas can be found, for example, \citet{Cherubini} and \citet{Nelson}.

Any $H\in \mathcal{F}(F_1, \cdots, F_d)$ satisfies
\[
 W(F_1(x_1), \cdots, F_d(x_d))\le H(\overrightarrow{x})\le M(F_1(x_1), \cdots, F_d(x_d)), \quad\hbox{for\, all}\quad \overrightarrow{x}\in\Real^d,
\]
where
\begin{equation}\label{intro.1}
W(\overrightarrow {u}):=\max\{u_1+\cdots+u_d-(d-1),0 \}\quad \hbox{and}\quad M(\overrightarrow {u}):=\min\{u_1,\cdots, u_d \},
\end{equation}
for $\overrightarrow{u}\in[0,1]^d $. $W$ and $M$ in \eqref{intro.1} are called the Fr\'{e}chet-Hoeffding lower and Fr\'{e}chet-Hoeffding upper bounds, respectively.
Note that  $M(F_1, \cdots, F_d)$  is a cumulative distribution of a $d$-variate random vector
while $W(F_1, \cdots, F_d)$ is not in general.
Let $\overline{H}$ be a survival distribution function defined as
\[
 \overline{H}(\overrightarrow{x}):=\P{X_1>x_1, \cdots, X_d>x_d}\quad\hbox{for}\quad \overrightarrow{x}\in \Real^d.
\]
 For $H, H^*\in \mathcal{F}(F_1, \cdots, F_d)$, the concordance ordering $H\prec H^*$ is defined by
\[
 H(\overrightarrow {x})\le H^*(\overrightarrow {x}) \quad \hbox{and}\quad \overline{H}(\overrightarrow {x})\le \overline{H^*}(\overrightarrow {x}) \quad\hbox{for all}\quad \overrightarrow {x}\in \Real^d.
\]
Furthermore, define $H= H^*$ if $$H(\overrightarrow{x})=H^*(\overrightarrow{x})$$ for any $\overrightarrow{x}\in \Real^d$.
Equivalently, denote $\overrightarrow{X}\eqd \overrightarrow{X^*}$ if $H= H^*$, where the cumulative distribution function of $\overrightarrow{X^*}$ is $H^*$.
Unless specified,
\[
 \overrightarrow{U}:=(U_1, \cdots, U_d),\quad\overrightarrow{U^*}:=(U_1^*, \cdots, U_d^*), \quad\hbox{and}\quad\overrightarrow{U^{**}}:=(U_1^{**}, \cdots, U_d^{**})
\]
are $d$-variate random vectors in $\mathcal{F}_d$ having copula $C$, $C^*$ and
$C^{**}$ as cumulative distributions functions, respectively. For example,
\[
 \P{U_1\le u_1, \cdots, U_d\le u_d}=C(\overrightarrow{u})
\]
for $\overrightarrow{u}\in [0,1]^d$.

It will be convenient to define the minimal and minimum copulas.
 For $d\ge 2$, we define $d$-dimensional copula $C\in \mathcal{F}_d$ as a minimum(maximal) copula if the inequality
 \[
  C^*\succ(\prec) C 
 \]
for any $d$-dimensional copula $C^*\in \mathcal{F}_d$.
Similarly, for $d\ge 2$, define $d$-dimensional copula $C\in \mathcal{F}_d$ as a minimal(maximal) copula if the inequality
 \[
  C^*\prec(\succ) C 
 \]
for some $d$-dimensional copula $C^*\in \mathcal{F}_d$ implies $C^*= C$.
Define the set of copulas $\mathbb{C}\subseteq \mathcal{F}_d$ to be {\it minimal in set concordance ordering} if any $C\in\mathbb{C}$ and $C^*\in\mathcal{F}_d$ with
\[
 C^*\prec C
\]
 implies
\[
 C^*\in \mathbb{C}.
\]
By definition, $\mathbb{C}$ is minimal in set concordance ordering if $\mathbb{C}$ is empty.
Clearly, the definition of minimality in set concordance ordering is a weaker concept than the definition of minimal copula.
In the minimality  of set concordance ordering, the quality of the minimality depends on the size of the set.
For example, Fr\'echet space is minimal in set concordance ordering.
On the other hand, if $\mathbb{C}$ has a single element, the definition of the minimality in set concordance ordering coincides with the definition of the minimal copula.

\subsection{Review of $d$-Countermonotonicity}
Comonotonicity has gained popularity in actuarial science and finance.
Conceptually, a random vector $\overrightarrow {X}$ is comonotonic if all of its components move in the same direction.
Comonotonicity is useful in several areas, such as the bound problems of an aggregate sum \citep{Goovaerts3, Cheung3} and hedging problems \citep{Cheung2}.
Recently, comonotonicity has been used in describing the economic crisis \citep{Dhaene2, Dhaene3, Ahnhix}.

Countermonotonicity is the opposite concept to comonotonicity.
Conceptually, in the bivariate case, a random vector $\overrightarrow {X}$ is countermonotonic if two components move
in the opposite directions.
The following classical results summarize
the equivalent conditions of countermonotonicity in bivariate dimensions.
\begin{definition}\label{def.6}
 A set $A\subset \Real^2$ is countermonotonic(comonotonic) if the following inequality holds
 \[
  (x_1-y_1)(x_2-y_2)\le(\ge) 0 \quad\hbox{for all}\quad {\overrightarrow{x}},\,{\overrightarrow{y}}\in \Real^2.
 \]
$\overrightarrow{X}$ is called countermonotonic(comonotonic) if it has countermonotonic(comonotonic) support.
\end{definition}

\begin{theorem}\label{thm.counter.1}
 For a bivariate random vector $\overrightarrow{X}$, we have the following equivalent statements.
 \begin{enumerate}
  \item[i.] $\overrightarrow{X}$ is countermonotonic
  \item[ii.] For any $\overrightarrow{x}\in \Real^2$
  \begin{equation}\label{eq.counter.4}
    \P{{\overrightarrow{X}}\le \overrightarrow{x}}=\max\left\{F_1(x_1)+F_2(x_2)-1,0 \right\}
  \end{equation}
  \item[iii.] For $\hbox{Uniform\rm($0,1$)}$ random variable $U_1$, we have
  \[
   {\overrightarrow{X}}\eqd \left(F_1^{-1}(U_1), F_2^{-1}(1-U_1) \right).
  \]

 \end{enumerate}

\end{theorem}

While the extension of comonotonicity into multivariate dimensions $d>2$ is straightforward, there is no obvious extension of countermonotonicity
into multivariate dimensions  $d>2$.
As discussed in \citet{Ahn7}, the difficulty of the extension of countermonotonicity arises partially due to the lack of minimum copula.
In this paper, we provide a set of minimal copulas, which can be viewed as a natural extension of countermonotonicity in two dimension into
multivariate dimensions.

\section{Weighted Countermonotonicity}\label{sec.3}

As an extension of countermonotonicity or negative extreme dependence in multivariate dimensions, \citet{Ahn7}
introduced the concept of $d$-CM.
While $d$-CM copulas are theoretically interesting, the existence and construction of $d$-CM copulas
with certain parametric functions remain unknown, and it may therefore be hard to apply $d$-CM copulas to various optimization problems.
Furthermore, the concept of $d$-CM may be too general
to describe the negative dependence concept as briefly specified in \citet{Ruodu3},
where the example of $(V, V, \cdots, V, 1-V)$ was given.
Alternatively, \citet{Ahn7} introduced the concept of strict $d$-CM as a special case of $d$-CM,
which is useful in some minimization problems.
However, because of the symmetricity of strict $d$-CM,
it cannot be used for non-symmetric optimization problems, as will be explained in Section \ref{optim}.
For completeness in the paper, we have summarized the definitions and properties of (strict) $d$-CM in the Appendix.

In this section, we introduce a new class of extremal negative dependent copulas,
which will be called $\overrightarrow{w}$-Countermonotonic ($\overrightarrow{w}$-CM) copulas
and can be interpreted as a set of minimal copulas as shown in Corollary \ref{rem.2} below.
Remark \ref{remark1} addresses
that the set of $\overrightarrow{w}$-CM copulas can be interpreted as generalized strict $d$-CM, and
further shows that $\overrightarrow{w}$-CM copulas are the subset of $d$-CM copulas.

\begin{definition}\label{def.w.1}
A $d$-variate random vector $\overrightarrow{X}$ is  $\overrightarrow{w}$-CM if
\begin{equation*}
\P{\sum\limits_{i=1}^d w_i\, F_i(X_i)=\frac{\sum_{i=1}^{d}w_i}{2}}=1.
\end{equation*}
Equivalently, we say that $H$ is $\overrightarrow{w}$-CM if $\overrightarrow{X}$ is $\overrightarrow{w}$-CM.
Furthermore, when $\overrightarrow{X}$ is $\overrightarrow{w}$-CM, we define $\overrightarrow{w}$ as a shape vector of $\overrightarrow{X}$.
\end{definition}

$\overrightarrow{w}$-CM can be regarded as multivariate extension of countermonotonicity into multivariate dimensions.
First, assume that $\overrightarrow{X}\in\mathcal{F}_2(F_1, F_2)$ is countermonotonic. Then, since $\overrightarrow{X}$ is a continuous random vector,
Theorem \ref{thm.counter.1}. iii concludes that
\[
\begin{aligned}
 F_1(X_1)+F_2(X_2) &\eqd F_1\circ F_1^{-1}(U_1) +F_2\circ F_2^{-1}(1-U_1) \\
\end{aligned}
\]
which in turn implies $\overrightarrow{X}$ is $\overrightarrow{w}$-CM for any $w_1=w_2>0$.
On the other hand, assume that $\overrightarrow{X}$ is $\overrightarrow{w}$-CM with $w_1=w_2>0$. Then by Definition \ref{def.w.1}, we have
 \[
  F_1(X_1)+F_2(X_2)=1
 \]
with probability $1$, which in turn concludes that the support of $\overrightarrow{X}$ is countermonotonic.
So we can conclude that
$\overrightarrow{X}\in\mathcal{F}_2(F_1, F_2)$ is countermonotonic if and only if $\overrightarrow{X}$ is $\overrightarrow{w}$-CM with $w_1=w_2$.

As can be expected from Definition \ref{def.w.1}, $\overrightarrow{w}$-CM is a property of copula only, and this is summarized in the following lemma.
The proof is similar to that of Lemma 1 in \citet{Ahn7}.
However, the result in the following lemma is more useful as it shows that the shape vector is invariant to marginal distributions.

\begin{lemma}
 Let $\bf X$ and ${\bf X}^*$ be random vectors from the distribution functions
 \[
  H=C(F_1, \cdots, F_d) \quad \hbox{and}\quad H^*=C(F_1^*, \cdots, F_d^*),
 \]
respectively, where marginal distribution functions, $F_1, \cdots, F_d$, are possibly different from marginal distribution functions, $F_1^*,\cdots, F_d^*$.
Then $\bf X$ is $\overrightarrow{w}$-CM if and only if ${\bf X}^*$ is $\overrightarrow{w}$-CM.
\end{lemma}
\begin{proof}
Since two random vectors $(F_1(X_1),\cdots, F_d(X_d) )$ and $(F_1^*(X_1^*),\cdots, F_d^*(X_d^*) )$
have copula $C$ as the same distribution functions, we have
\begin{equation}
\P{\sum\limits_{j=1}^d w_1\, F_i(X_i)=\frac{\sum_{i=1}^{d}w_i}{2}}=1,
\end{equation}
if and only if
\begin{equation}
\P{\sum\limits_{j=1}^d w_1\, F_i^*(X_i^*)=\frac{\sum_{i=1}^{d}w_i}{2}}=1.
\end{equation}
Hence we conclude that
$\bf X$ is $\overrightarrow{w}$-CM if and only if ${\bf X}^*$ is $\overrightarrow{w}$-CM.
\end{proof}

In the following definition, we provide the copula version of $\overrightarrow{w}$-CM.
Note that, for the property of $\overrightarrow{w}$-CM,
it is enough to study the copula version of $\overrightarrow{w}$-CM, because $\overrightarrow{w}$-CM is a property of copula only.
Hence, throughout this paper, we will use the following definition as the definition of $\overrightarrow{w}$-CM.
\begin{definition}\label{new.def}
A $d$-variate random vector $\overrightarrow{U}$ is  $\overrightarrow{w}$-CM if
\begin{equation}\label{eq.strict}
\P{\sum\limits_{i=1}^d w_i\, U_i=\frac{\sum_{i=1}^{d}w_i}{2}}=1.
\end{equation}
Equivalently, we say that $C$ is $\overrightarrow{w}$-CM if $\overrightarrow{U}$ is $\overrightarrow{w}$-CM. 
Here, $\overrightarrow{w}$ is called as a shape vector of $\overrightarrow{X}$.
\end{definition}

\begin{remark}\label{remark1}
Note that $\overrightarrow{w}$-CM is $d$-CM with parameter functions
\[
 f_i(y)=c_i F_i(y)
\]
for $i\in\{1, \cdots, d\}$ and $y\in \Real$, where
$$c_i:=\frac{2\,w_i}{\sum\limits_{j=1}^{d}w_j}.$$
Furthermore, since $\overrightarrow{w}$-CM coincides with strict $d$-CM when $w_1=\cdots=w_d$,
the set of strict $d$-CM copula is the subset of $\overrightarrow{w}$-CM copulas.
in the Appendix.
For convenience, we summarize the definitions of $d$-CM and strict $d$-CM in Definition \ref{def.multi3} and Definition \ref{def.multi4} in the Appendix.
\end{remark}

The following corollary explains that the set of $\overrightarrow{w}$-CM copulas can be regarded
to have minimality in concordance ordering as a set.
Since $\overrightarrow{w}$-CM is a special case of $d$-CM as shown in Remark \ref{remark1}, the proof of the following corollary is immediate from
\citet{Ahn7}. However, for completeness in the paper, we present the proof in the Appendix.

\begin{corollary}\label{rem.2}
For given $\overrightarrow{w}\in\Real_+^d$, let $\mathbb{C}$ be the set of $\overrightarrow{w}$-CM: i.e. $\mathbb{C}$ is defined as
$$\mathbb{C}:=\left\{ C\in \mathcal{F}_d\big\vert \hbox{ $C$ is $\overrightarrow{w}$-CM}\right\}.$$
Then $\mathbb{C}$ is minimal in set concordance ordering.

\end{corollary}

As briefly mentioned in Section \ref{sec.1}, since there is no minimum copula available for $d\ge 3$,
it is clear that minimal copulas will play a key role in various minimization problems.
In this sense, Corollary \ref{rem.2} addresses an important property of $\overrightarrow{w}$-CM copulas:
the set of $\overrightarrow{w}$-CM copulas achieves minimality in the sense that
there are no copulas strictly smaller than the $\overrightarrow{w}$-CM copula other than $\overrightarrow{w}$-CM copulas.
Hence, the concept of $\overrightarrow{w}$-CM can be useful in various minimization/maximization problems
as will be explained in Section \ref{optim} below. For a discussion of the usage of $\overrightarrow{w}$-CM copulas, it is essential to check
the existence of $\overrightarrow{w}$-CM copulas as will be shown in the next section.

\section{Condition to Achieve the Weighted Countermonotonicity}\label{sec.cond}

Depending on the given marginal distributions, \eqref{eq.strict} may not be always achieved. For example, for
$(w_1, w_2)=(2,1)$, none of $(U_1, U_2)\in \mathcal{F}_2$
can achieve the condition in \eqref{eq.strict}.
In this section, we provide the equivalence condition of the weight $\overrightarrow{w}\in\Real_+^d$
for the existence of $\overrightarrow{w}$-CM copula.
Note that a similar result can be found in a recent working paper by \citet{Ruodu4},
which explains an algebraic way of constructing $\overrightarrow{w}$-CM copulas.
We first define the set of weights where $\overrightarrow{w}$-CM copulas exists.

\begin{notation}
Define the set of weights in $3$-dimensions as follows
\[
\mathbb{W}_3:=\left\{ (w_1, w_2, w_3)\in \Real_+^3 \big\vert\sum\limits_{i=1}^{3}w_i \ge 2\max\{w_1, w_2,w_3\}\right\} .
\]

\end{notation}

Note that the set $\mathbb{W}_3$ is equivalent with the set of the line lengths in triangles (including degenerate triangles).

\begin{lemma}\label{ex.thm0}
 For any $\overrightarrow{w}\in \mathbb{W}_3$, there exists $\overrightarrow{w}$-CM copulas.

\end{lemma}
\begin{proof}
 For convenience, define
 \begin{equation}\label{zs}
  z_1:=\frac{w_2+w_3-w_1}{2w_2},\;\, z_2:=\frac{w_3+w_1-w_2}{2w_3}\quad \hbox{and} \;\, z_3:=\frac{w_1+w_2-w_3}{2w_1},\;\,\hbox{for some}\;\, (w_1, w_2, w_3)\in \Real_+^3.
 \end{equation}
 and denote $\overrightarrow{u}\in \Phi(\overrightarrow{w})$ if
 \[
  u_1\,w_1+u_2\,w_2+u_3\,w_3=\frac{w_1+w_2+w_3}{2}.
 \]
Now, let us consider the following three points
\[
  \overrightarrow{p_1}:=(1,z_1,0),\quad \overrightarrow{p_2}:=(0,1,z_2)\quad \hbox{and} \quad \overrightarrow{p_3}:=(z_3,0,1),
\]
and observe that the points satisfy $\overrightarrow{p_i}\in \Phi(\overrightarrow{w})$ for $i\in\{1,2,3\}$.
Hence any point on the line that connects $\overrightarrow{p_i}$ and $\overrightarrow{p_j}$ is again in $\Phi(\overrightarrow{w})$: i.e.
\begin{equation}\label{my.star}
 t\overrightarrow{p_i}+(1-t)\overrightarrow{p_j}\in \Phi(\overrightarrow{w})
\end{equation}
for any $0\le t\le 1$ and $i,j\in\{1,2,3\}$.
Further, by the assumption $\overrightarrow{w}\in\mathbb{W}_3$, the following inequalities can be derived
\[
 0\le z_i\le 1 \quad\hbox{for}\quad i\in\{1,2,3\},
\]
which in turn implies
\begin{equation}\label{bp}
t\overrightarrow{p_i}+(1-t)\overrightarrow{p_j}\in [0,1]^3
\end{equation}
for any $0\le t\le 1$ and $i,j\in\{1,2,3\}$.
Note that the trace of \eqref{bp} is triangular in $[0,1]^3$ with vertices lying on $\overrightarrow{p_1}$, $\overrightarrow{p_2}$ and $\overrightarrow{p_3}$.

Now for the given triangle with vertices $\overrightarrow{p_1}$, $\overrightarrow{p_2}$ and $\overrightarrow{p_3}$,
we give positive weights $m_{1,2}$, $m_{2,3}$, $m_{3,1}$ to each edge $\overline{p_1p_2}$, $\overline{p_2p_3}$ and $\overline{p_3p_1}$
such that weights are uniformly distributed on each edge.
Here we assume that that the sum of weights is given as
$m_{1,2}+m_{2,3}+m_{3,1}=1$,
so that the weights on the edges of the triangle define
a random vector $\overrightarrow{X}=(X_1, X_2, X_3)$.
In defining $H$ as the cumulative distribution function of the random vector $\overrightarrow{X}$,
our goal is to show that there exist the weights $(m_{1,2},\,m_{2,3},\,m_{3,1})\in\mathbb{R}_+^3$ which make $H$ a copula.

To show that $H$ is a copula, it is enough to show that $H$ is $2$-increasing and that the marginals of $H$ are a uniform$[0,1]$ distribution \citep{Nelson}.
Since $H$ is defined by the nonnegative weights $m_{12}, m_{23}, m_{31}$ that are distributed on the edges of the triangular,
it is obvious that $H$ is $2$-increasing.
Now, it remains to show that marginals of $H$ are a uniform$[0,1]$ distribution.
Since weights $m_{12}, m_{23}, m_{31}$ are uniformly distributed on each edge, it is enough
to check uniformity on each vertex of the triangle, which is equivalent to show
\begin{equation}\label{3eq}
 \P{X_1\le z_1}=z_1,\quad\P{X_2\le z_2}=z_2\quad\hbox{and}\quad \P{X_3\le z_3}=z_3.
\end{equation}
Each equation in \eqref{3eq} is equivalent with
\begin{equation}
 \begin{aligned}\label{3eq2}
   z_1&= m_{12}\cdot 0+m_{23} \,z_1  + m_{31}\cdot 1,\\
  z_2&= m_{12}\cdot 1 + m_{23} \cdot 0 + m_{31}\,z_2 ,\\
  z_3&= m_{12}\,z_3 + m_{23}\cdot 1  + m_{31}\cdot 0 ,\\
 \end{aligned}
\end{equation}
respectively.
The solution of \eqref{3eq2} is
\begin{equation}\label{3sol}
 \begin{aligned}
 m_{12}=\frac{(1-z_1)(1-z_2)(1-z_3)-(1-z_1)(1-z_2)+(1-z_1)} {(1-z_1)(1-z_2)(1-z_3)+1},\\
 m_{23}=\frac{(1-z_1)(1-z_2)(1-z_3)-(1-z_2)(1-z_3)+(1-z_2)} {(1-z_1)(1-z_2)(1-z_3)+1},\\
 m_{31}=\frac{(1-z_1)(1-z_2)(1-z_3)-(1-z_3)(1-z_1)+(1-z_3)} {(1-z_1)(1-z_2)(1-z_3)+1}.
 \end{aligned}
\end{equation}
While $m_{1,2}+m_{2,3}+m_{3,1}\neq 1$ for general $(z_1, z_2, z_3)\in[0,1]^3$,
a tedious but straightforward calculation shows that, with $(z_1,\,z_2,\,z_3)$ defined in \eqref{zs},
$(m_{1,2}, m_{2,3}, m_{3,1})$ defined in \eqref{3sol} always satisfies $$m_{1,2}+ m_{2,3}+ m_{3,1}=1.$$
Finally, \eqref{my.star} derives that $\overrightarrow{X}$ is $\overrightarrow{w}$-CM.

\end{proof}

 While $(z_1,z_2,z_3)$ is some vector in $[0,1]^3$, it is worth mentioning that the definition \eqref{zs} is crucial
 to guarantee that the solution $(m_{1,2}, m_{2,3}, m_{3,1})$ of \eqref{3eq} satisfies
 \[
  m_{1,2}+ m_{2,3}+ m_{3,1}=1.
 \]
In other words, $(m_{1,2}, m_{2,3}, m_{3,1})$, which satisfies \eqref{3eq} may not satisfy
 \[
  m_{1,2}+ m_{2,3}+ m_{3,1}=1
 \]
for and arbitrarily given $(z_1, z_2, z_3)\in[0,1]^3$ that does not satisfy the condition \eqref{zs}.
For example, for the arbitrarily given $(z_1, z_2, z_3)=(0.5, 0.3, 0.2)$, the solution $(m_{1,2}, m_{2,3}, m_{3,1})$ of \eqref{3eq}  defined in \eqref{3sol} has
 \[
  m_{1,2}+ m_{2,3}+ m_{3,1}>1.
 \]
The following lemma is an extension of Lemma \ref{ex.thm0} into multivariate dimensions $d\ge3$.
\begin{lemma}\label{ex.thm}
For given $\overrightarrow{w}\in\Real_+^d$, if there exist disjoint subsets $A$, $B$ and $C$ of $\{w_1, \cdots, w_d\}$ such that
 \begin{equation}\label{cond1}
  \left(\sum\limits_{w_i\in A} w_i, \sum\limits_{w_i\in B} w_i, \sum\limits_{w_i\in C} w_i \right) \in\mathbb{W}_3\quad\hbox{and}\quad A\cup B\cup C =\left\{w_1, \cdots, w_d \right\},
 \end{equation}
then there exists a random vector $\overrightarrow{U}$ whose marginals are uniform[$0,1$] and it satisfies
 \begin{equation}\label{cond2}
 \sum\limits_{i=1}^{d}w_iU_i=\frac{\sum_{i=1}^{d}w_i}{2}.
 \end{equation}
\end{lemma}
\begin{proof}
Let $\overrightarrow{U}$ be random vectors with marginals being uniform[0,1].
Further, let
\[
\begin{aligned}
 V_1:=\sum\limits_{i\in A} U_i, \quad
 V_2:=\sum\limits_{i\in B} U_i \quad\hbox{and}\quad
 V_3:=\sum\limits_{i\in C} U_i.
\end{aligned}
\]
Now the proofs are trivial if we set $U_i$'s in the same subset as being comonotonic
i.e. $U_i$ and $U_j$ are comonotonic if either $i,j\in A$, $i,j\in B$ or $i,j\in C$.
\end{proof}

 The following lemma provides the equivalence condition of \eqref{cond1}, which is more intuitive and easy to verify.
\begin{lemma}\label{contra}
For the given weight $\overrightarrow{w}\in\Real_+^d$, we have the following inequality
\begin{equation}\label{contra.1}
\max\{w_1, \cdots, w_d\}\le \sum\limits_{i=1}^d w_i-\max\{w_1, \cdots, w_d\}
\end{equation}
if and only if there exist disjoint subsets $A$, $B$ and $C$ of $\{w_1, \cdots, w_d\}$ satisfying \eqref{cond1}.
\end{lemma}
\begin{proof}
First observe that \eqref{cond1} implies \eqref{contra.1} is trivial. Hence it remains to show \eqref{contra.1} implies \eqref{cond1}.
Without loss of generality, let $w_1\ge\cdots\ge w_d$.
 For any integer $d\ge 3$, define
 \[
  w_2^*:=\sum\limits_{i\in \mathbb{Z}_O}w_i \quad\hbox{and}\quad w_3^*:=\sum\limits_{i\in \mathbb{Z}_E}w_i
 \]
 where $\mathbb{Z}_O:=\{w_i\big\vert i\neq 1 \hbox{ and $i\le d$ is odd number} \}$ and
 $\mathbb{Z}_E:=\{w_i\big\vert \hbox{$i\le d$ is even number} \}$.
 Then it is straightforward to show that
 \begin{equation}\label{contra.2}
  w_1+w_2^*\ge w_3^* \quad\hbox{and}\quad w_1+w_3^*\ge w_2^*
 \end{equation}
 Hence, along with \eqref{contra.2}, if we assume that $\overrightarrow{w}$ satisfies \eqref{contra.1}, we can conclude $(w_1, w_2^*, w_3^*)\in \mathbb{W}_3$,
 which in turn implies \eqref{cond1} with $A=\{w_1\}$, $B=\{w_i\big\vert i\in\mathbb{Z}_O \}$ and $C=\{w_i\big\vert i\in\mathbb{Z}_E \}$.
\end{proof}

So far in Lemma \ref{ex.thm} and Lemma \ref{contra}, we have provided sufficient conditions for the existence
$\overrightarrow{w}$-CM copula.
Then the natural question is to check whether they are also necessary conditions or not. The following corollary shows that
the condition in \eqref{contra.1} is also a necessary condition for the existence of $\overrightarrow{w}$-CM copulas.

\begin{corollary}\label{cor.wcm}
For the given weight $\overrightarrow{w}\in\mathbb{R}_+^d$, there exists random vector $\overrightarrow{w}$-CM random vector $\overrightarrow{U}$
if and only if $\overrightarrow{w}$ satisfies
\begin{equation}\label{cond.min}
\max\{w_1, \cdots, w_d\}\le \sum\limits_{i=1}^d w_i-\max\{w_1, \cdots, w_d\}.
\end{equation}
\end{corollary}
\begin{proof}
It is enough to show that $\overrightarrow{w}$-CM implies \eqref{cond.min}.
 First, consider a weight $(w_1, \cdots, w_d)\in\Real_+^d$ such that
one weight, say $w_1$, is greater than the sum of all other weights
\begin{equation}\label{cond.min1}
w_1> \sum\limits_{i=2}^d w_i.
\end{equation}
Then, it is obvious that there does not exist any random vector $\overrightarrow{U}$ whose marginals are uniform[$0,1$] and satisfies \eqref{eq.strict}:
this can be easily verified using the following variance comparison;
\[
 \Var{w_1U_1} > \Var{\sum\limits_{i=2}^{d}w_iU_i}.
\]
Hence, we can conclude that there does not exist $\overrightarrow{w}$-CM copula under the condition \eqref{cond.min1}, which concludes the claim.

\end{proof}

   \begin{figure}[h!]
    \centering
      \includegraphics[width=0.49\textwidth]{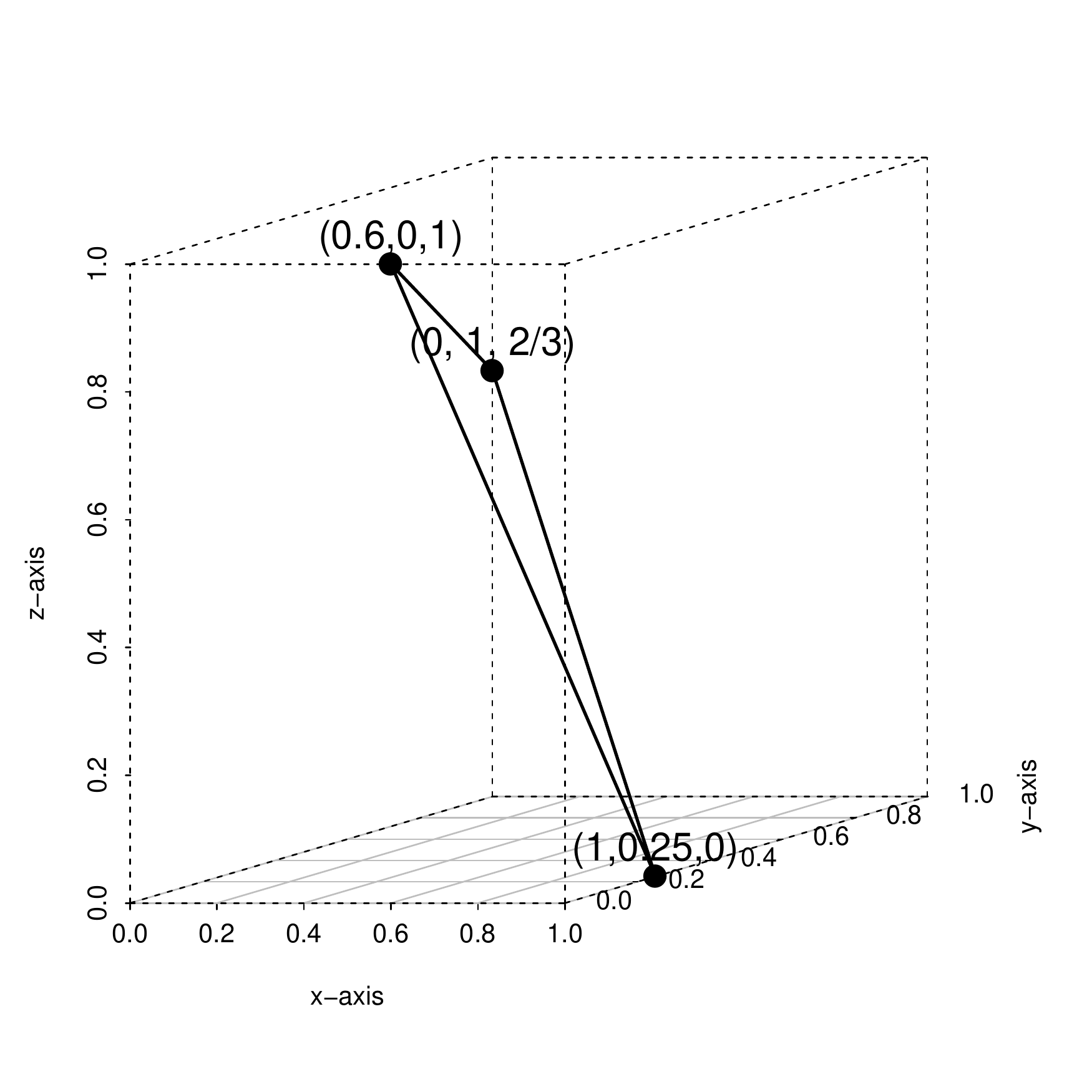}
      \caption{$w$-CM Copula with $(w_1, w_2, w_3)=(5,4,3)$}
      \label{figure.100}
  \end{figure}

 \begin{remark}\label{exit.rem}
  For any $\overrightarrow{w}\in \Real_+^d$ satisfying \eqref{cond.min}, the choice of $\overrightarrow{w}$-CM copula is not unique.
For example, in the proof of Lemma \ref{ex.thm0}, we show how to construct $\overrightarrow{w}$-CM copula $C$ for $d=3$.
On the other hand, the construction method in the proof of Lemma \ref{ex.thm0} and the following choices of $z_1^*,\,z_2^*,\, z_3^*\in [0,1]$ defined as
 \begin{equation*}
  z_1^*:=\frac{w_2+w_3-w_1}{2w_3},\;\, z_2^*:=\frac{w_3+w_1-w_2}{2w_1}\quad \hbox{and} \;\, z_3^*:=\frac{w_1+w_2-w_3}{2w_2},
 \end{equation*}
and three points $\overrightarrow{p_1^*},\, \overrightarrow{p_2^*},\, \overrightarrow{p_3^*}\in [0,1]^3$ defined as
\[
  \overrightarrow{p_1^*}:=(1,0,z_1),\quad \overrightarrow{p_2^*}:=(z_2,1,0)\quad \hbox{and} \quad \overrightarrow{p_3^*}:=(0,z_3,1),
\]
will derive another choice of $\overrightarrow{w}$-CM copula $C^*$ with $C\neq C^*$.

 \end{remark}

In the bivariate case, Corollary \ref{cor.wcm} concludes that $(U_1, U_2)$ being $\overrightarrow{w}$-CM implies that $w_1=w_2$
which coincides with the concept of countermonotonicity as we already mentioned in Section \ref{sec.3}.
The following example shows the numerical example of the construction of $\overrightarrow{w}$-CM copula
using the logic in the proof of Lemma \ref{ex.thm0}.

\begin{example}
Let $(w_1, w_2, w_3)=(5,4,3)$.
Since $$5=w_1\le w_2+w_3=4+3,$$ we know that $\overrightarrow{w}\in\mathcal{W}_3$ and, by Corollary \ref{cor.wcm},
there exists a $\overrightarrow{w}$-CM random vector $\overrightarrow{U}\in\mathcal{F}_3$. Using the techniques used in
\eqref{ex.thm0}, we can construct $\overrightarrow{w}$-CM random vector $\overrightarrow{U}$ having mass
$m_{12}$, $m_{23}$ and $m_{31}$ uniformly distributed on the each edge $\overrightarrow{p_1p_2}$, $\overrightarrow{p_2p_3}$ and
$\overrightarrow{p_3p_1}$, respectively. Here
\[
 \overrightarrow{p_1}=\left(1, \frac{1}{4},0\right), \quad \overrightarrow{p_2}=\left(0, 1,\frac{2}{3}\right)\quad \hbox{and}\quad \overrightarrow{p_3}=\left(\frac{3}{5} ,0 ,1\right)
\]
and
\begin{equation}\label{mass.m}
 m_{12}=\frac{6}{11},\quad m_{23}=\frac{3}{11}\quad\hbox{and}\quad m_{31}=\frac{2}{11}
\end{equation}
Hence, for example, we have
\[
\begin{aligned}
 C\left(1, 0.25,1/3\right)&=1/3*m_{31}\\
 &=\frac{2}{33}.
\end{aligned}
\]
Finally, Figure \ref{figure.100} shows the support of random vector $\overrightarrow{U}$.

\end{example}

\section{Application to the Variance Minimization Problem}\label{optim}

Finding the maximum and the minimum of variance in the aggregated sum with given marginal distributions is the classical optimization problem
\begin{equation}\label{var.eq}
 \Var{\sum\limits_{i=1}^{d}X_i},\quad\hbox{for given $X_i\sim F_i$}, \quad i=1, \cdots, d.
\end{equation}
First of all, the maximization of \eqref{var.eq} is straightforward using the comonotonic random vectors.
For the minimization problem with $d=2$, the answer is trivial with countermonotonic random variables.
Regarding general dimensions $d\ge 3$, minimization of \eqref{var.eq} was solved for some cases of marginal distributions  \citep{Ruschendorf, Ruschendorf2, Ruodu, Ruodu3}.
However, minimization of \eqref{var.eq} is not easy in general for $d\ge3$.
The following remark, which can be easily derived from Theorem 2.7 of \citet{Dhaene3}, states that variance minimization problems are related with concordance ordering,
which may offer some hints in the minimization of \eqref{var.eq}.

\begin{remark}\label{order.var}
Let $F_1, \cdots, F_d$ be distribution functions having finite variances.
 If $$\overrightarrow{X^*}, \, \overrightarrow{X}\in\mathcal{F}_d(F_1, \cdots, F_d)$$ with $C^*\prec C$, then
 \[
  \Var{\sum\limits_{i=1}^{d}X_i^*}\le  \Var{\sum\limits_{i=1}^{d}X_i}.
 \]
\end{remark}

From the remark, it is clear that the minimization and maximization of \eqref{var.eq}
is related with a minimum and maximum copula.
While maximization of \eqref{var.eq} is related with the comonotonic copula,
due to the absence of the minimum copula for $d\ge3$, the minimization of \eqref{var.eq} is related with
the set of the minimal copulas.
Of course, the choice of the proper set of minimal copulas depends on the marginal distributions.
Among many other choices of the marginal distributions in \eqref{var.eq},
this paper considers the uniform marginal distributions as shown in the following definition, which may be the simplest versions of \eqref{var.eq}.
The following assumption is useful to simplify the notation in several theorems in this section.
\begin{assumption}\label{assm.1}
Assume that $w_1= \max\{w_1, \cdots, w_d\}$.
\end{assumption}

\begin{definition}\label{var.min}
 For given $\overrightarrow{w}\in\Real_+^d$, define
 \begin{equation}\label{min.eq}
  m_-(\overrightarrow{w}):=\inf\left\{ \Var{\sum\limits_{i=1}^{d} \widetilde{U_i}}\bigg\vert \,\widetilde{U_i}\,\,\hbox{is uniform$[0,w_i]$ random variables, }\,i=1,\cdots,d\right\}
 \end{equation}
 and
  \[
  m_+(\overrightarrow{w}):=\sup\left\{ \Var{\sum\limits_{i=1}^{d} \widetilde{U_i}}\bigg\vert \,\widetilde{U_i}\,\,\hbox{is uniform$[0,w_i]$ random variables, }\,i=1,\cdots,d\right\}.
 \]
where $\widetilde{U_i}\,\,\hbox{is uniform$[0,w_i]$ random variables for }i=1,\cdots,d$.

\end{definition}
Equivalently, $m_-(\overrightarrow{w})$ and $m_+(\overrightarrow{w})$ can be written as
 \[
  m_-(\overrightarrow{w})=\inf\left\{ \Var{\sum\limits_{i=1}^{d} w_iU_i}\bigg\vert \,\overrightarrow{U}\in \mathcal{F}_d\right \}
 \]
 and
  \[
  m_+(\overrightarrow{w})=\sup\left\{ \Var{\sum\limits_{i=1}^{d} w_iU_i}\bigg\vert \,\overrightarrow{U}\in \mathcal{F}_d\right\}.
 \]
The upper bound
$$m_+(\overrightarrow{w})=\frac{1}{12}\left(\sum\limits_{i=1}^{d}w_i \right)^2$$
is achieved if and only if $\overrightarrow{U}$
is comonotonic \citep{Kaas, Dhaene5, Dhaene}.
Regarding the lower bound, when
\begin{equation}\label{lower.c}
2\,\max\{w_1, \cdots, w_d\} \le \sum\limits_{i=1}^{d} w_i,
\end{equation}
Corollary \ref{cor.wcm} concludes that
\begin{equation*}
m_-(\overrightarrow{w})=0.
\end{equation*}
However for $\overrightarrow{w}$ which does not satisfy \eqref{lower.c}, minimization is not straightforward.

For $\overrightarrow{w}\in\Real_+^d$ which does not satisfy \eqref{lower.c},
Theorem \ref{lem.m8} below finds the explicit expression for $m_-({\overrightarrow{w}})$.
More importantly, we also show that the minimum $m_-({\overrightarrow{w}})$ is achieved with $\overrightarrow{w^*}$-CM copulas even though
$\overrightarrow{w^*}\in\Real_+^d$ may not be the same as $\overrightarrow{w}$.
Finally, Corollary \ref{lem.m1} provides the complete solution for $m_-({\overrightarrow{w}})$ and $m_+({\overrightarrow{w}})$ for any given
$\overrightarrow{w}\in\Real_+^d$.
Before we examine the main results, it is convenient to present the following lemma and notations.

\begin{lemma}\label{lem.m}
  Let $\overrightarrow{w}\in\Real_+^d$ satisfy Assumption \ref {assm.1} and
 \begin{equation}\label{my.cond5}
2\,w_1=\sum\limits_{i=1}^{d}w_i.
 \end{equation}
 Then the following inequality holds
 \begin{equation}\label{my.cond}
  \cov{U_1, \sum\limits_{i=1}^{d}w_iU_i }\ge 0,
 \end{equation}
where the equality holds if and only if $\overrightarrow{U}$ is $\overrightarrow{w}$-CM.
\end{lemma}
\begin{proof}
We first prove the inequality \eqref{my.cond} as
\begin{equation*}
 \begin{aligned}
  \cov{U_1, \sum\limits_{i=1}^{d}w_iU_i } &= \sum\limits_{i=1}^{d} \cov{U_1, w_iU_i }\\
                                          &= w_1\Var{U_1} + \sum\limits_{i=2}^{d} w_i {\rm corr}\left[ U_1, U_i \right] \sqrt{\Var{U_1}}\sqrt{\Var{U_i}}\\
                                          &\ge w_1\Var{U_1}-\sum\limits_{i=2}^{d} w_i \sqrt{\Var{U_1}}\sqrt{\Var{U_i}}\\
                                          &=\left[w_1-\sum\limits_{i=2}^{d} w_i \right]\Var{U_1}\\
                                          &=0
 \end{aligned}
\end{equation*}
where the inequality arises from the fact that correlation of any two random variables is greater than $-1$, and the last equality is from the condition
\eqref{my.cond5}.
Furthermore, since $\overrightarrow{w}$ satisfies the condition \eqref{cond.min}, Corollary \ref{cor.wcm}
concludes that the inequality in \eqref{my.cond} becomes equality if and only if
$\overrightarrow{U}$ is $\overrightarrow{w}$-CM.
\end{proof}

\begin{notation}\label{ref.ref}
For given $\overrightarrow{w}\in\Real_+^d$, define $\overrightarrow{w^*}:=(w_1^*, \cdots, w_d^*)$ as 
\begin{equation}\label{def.w}
 w_i^*
 :=\begin{cases}
                        w_i; & \hbox{ if }\;\; 2\, w_i\le \sum\limits_{j=1}^{d}w_j\\
                        \sum\limits_{j=1}^{d}w_i - w_i   ; & \hbox{ if }\;\; 2\, w_i> \sum\limits_{j=1}^{d}w_j
                       \end{cases}
\end{equation}
for $i=1,\cdots,d$. Then, one can easily confirm $\overrightarrow{w^*}\in\Real_+^d$.
Further, let
\[
\begin{aligned}
 l(\overrightarrow{w})
  &:=\frac{1}{12} \left[\left( 2\,\max\{w_1,\cdots, w_d\}-\sum\limits_{i=1}^{d}w_i\right)_+\right]^2.\\
\end{aligned}
\]

\end{notation}

Since there always exists $\overrightarrow{w}$-CM random vector $\overrightarrow{U}\in\mathcal{F}_d$ for
$\overrightarrow{w}\in\Real_+^d$, which satisfies the condition \eqref{cond.min},
we conclude $m_-(\overrightarrow{w})=0$ in this case.
The following proposition provides the tight bound of $m_-(\overrightarrow{w})$ when $\overrightarrow{w}$ does not
satisfy the condition \eqref{cond.min}. The main idea of the proof is to shrink the largest weight so that the new weights satisfy
the condition \eqref{cond.min}, which in turn results in constant summation or zero variance. Then Lemma \ref{lem.m} shows that only the remaining part of the largest weight contributes the lower bound of the variance specified in \eqref{min.eq}.
\begin{theorem}\label{lem.m8}
  Let the weight vector $\overrightarrow{w}\in\Real_+^d$ satisfy Assumption \ref{assm.1} and
  \begin{equation}\label{eq.pp}
  2\, w_1> \sum\limits_{i=1}^{d}w_i.
  \end{equation}
  Then the following inequality holds
 \begin{equation}\label{eq.m8}
  l(\overrightarrow{w})  \le \Var{\sum\limits_{i=1}^{d} w_iU_i} ,
 \end{equation}
where the inequality is attained if and only if\, $\overrightarrow{U}$ is $\overrightarrow{w^*}$-CM.
Furthermore, for any $\overrightarrow{w}\in \Real_+^d$ satisfying the condition \eqref{eq.pp}, there exists random vector
$\overrightarrow{U}\in\mathcal{F}_d$ which achieves the equality in \eqref{eq.m8}.

\end{theorem}
\begin{proof}
 First, observe
$w_1-w_1^*>0$
for the given $\overrightarrow{w}$ satisfying the condition \eqref{eq.pp}.
Further we have
\begin{equation}\label{lem.m10}
\begin{aligned}
 2w_1^*&= 2\left( \sum\limits_{i=1}^{d}w_i-w_1\right)\\
       &= \left( \sum\limits_{i=1}^{d}w_i-w_1\right)+\sum\limits_{i=2}^{d}w_i\\
       &= \sum\limits_{i=1}^{d}w_i^*.
\end{aligned}
\end{equation}
 Since $w_1-w_1^*>0$, we have
\begin{equation}\label{ahnb2}
\begin{aligned}
 \Var{\sum\limits_{i=1}^{d}w_iU_i}&=\Var{(w_1-w_1^*)U_1}+\Var{w_1^*U_1+\sum\limits_{i=2}^{d}w_iU_i}\\
 &\quad\quad\quad\quad+\cov{(w_1-w_1^*)U_1,\;\; w_1^*U_1+\sum\limits_{i=2}^{d}w_iU_i}\\
 &=\Var{(w_1-w_1^*)U_1}+\Var{w_1^*U_1+\sum\limits_{i=2}^{d}w_i^*U_i}\\
 &\quad\quad\quad\quad+(w_1-w_1^*)\cov{U_1,\;\; w_1^*U_1+\sum\limits_{i=2}^{d}w_i^*U_i}\\
 &\ge \Var{(w_1-w_1^*)U_1}\\
 &=\frac{1}{12}(w_1-w_1^*)^2,
\end{aligned}
\end{equation}
where the last inequality is from \eqref{lem.m10} and  Lemma \ref{lem.m}. 
Furthermore, since variance of the random variable is $0$ if and only if the random variable is constant with probability $1$, Lemma \ref{lem.m} concludes that
the  inequality in \eqref{ahnb2} is equality if and only if $\overrightarrow{U}$ is
$\overrightarrow{w^*}$-CM.

Finally, since $\overrightarrow{w^*}$ satisfies the condition \eqref{lem.m10} (hence the condition \eqref{cond.min}),
there always exists $\overrightarrow{w^*}$-CM random vector
$\overrightarrow{U}\in \mathcal{F}_d$, which in turn implies that the equality in \eqref{eq.m8} can be always achieved.

\end{proof}

Based on Theorem \ref{lem.m8},
the following corollary provides the complete solution for the optimization problem in Definition \ref{var.min}.
\begin{corollary}\label{lem.m1}
  For the given $\overrightarrow{w}\in\Real_+^d$,
  the following inequality holds
 \begin{equation}\label{eq.m}
  l(\overrightarrow{w})  \le \Var{\sum\limits_{i=1}^{d} w_iU_i}
  \le \frac{1}{12}\left( \sum\limits_{i=1}^{d}w_i\right)^2.
 \end{equation}
The lower bound of \eqref{eq.m} is attained if and only if\, $\overrightarrow{U}$ is $\overrightarrow{w^*}$-CM and
the upper bound of \eqref{eq.m} is attained if and only if\, $\overrightarrow{U}$ is comonotonic.
\end{corollary}
\begin{proof}

 The upper bound in \eqref{eq.m} is a classical result which can explained by comonotonic $\overrightarrow{U}$; see \citet{Kaas, Dhaene5, Dhaene} for details.
 Hence it is enough to show the lower bound
  \begin{equation}\label{eq.m10}
  \frac{1}{12} \left[\left( 2\,\max\{w_1, \cdots, w_d\}-\sum\limits_{i=1}^{d}w_i\right)_+\right]^2 \le \Var{\sum\limits_{i=1}^{d} w_iU_i}
 \end{equation}
 and the equality holds if and only if $\overrightarrow{U}$ is $\overrightarrow{w^*}$-CM.

 For the lower bound in \eqref{eq.m}, consider two cases depending on $\overrightarrow{w}$.
First, consider the following condition on $\overrightarrow{w}\in\Real_+^d$
\[
 2\, \max\{w_1, \cdots, w_d\}- \sum\limits_{i=1}^{d}w_i \le 0.
\]
Since $l(\overrightarrow{w})=0$
in this case, nonnegativeness of the variance shows the left inequality of \eqref{eq.m}.
Furthermore, Corollary \ref{cor.wcm} implies that
the equality in \eqref{eq.m} holds if and only if $\overrightarrow{U}$ is $\overrightarrow{w^*}$-CM.
Finally, for $\overrightarrow{w}$ satisfying
\[
 2\, \max\{w_1, \cdots, w_d\}- \sum\limits_{i=1}^{d}w_i > 0,
\]
the same result was summarized in Theorem \ref{lem.m8}.
\end{proof}

\section{Marginal Free Herd Behavior Index}\label{hix.section}

In this section, we define the marginal free measure of dependence which can be interpreted as the measure for the herd behavior.
We first review various herd behavior indices, which measure the degree of comovement or comonotonicity \citep{Dhaene2,Daniel14,Ahnhix}.
While such indices need to be marginal free
(because the concept of comovement or comonotonicity is a definition of copula only),
in Subsection \ref{sub.marginal}, we observe that such indices can be distorted by marginal distributions.
Alternatively, Subsection \ref{def.six} presents a definition of measures of dependence that is free of marginal distributions.

\subsection{Review of the Measures of Dependence}
Herd behavior is a general concept often used in various fields such as financial and psychology
to describe the irrational comovement of members in a group.
The recent financial crises have further highlighted the importance of understanding the herd behavior.
There have been several attempts to measure the herd behaviors through herd behavior indices.
In this subsection, we will briefly review some known herd behavior indices in the financial context.

Let ${\overrightarrow{X}}$ be $d$ individual stock prices at a time $t$ assuming that the current time is fixed at $0$.
For the given ${\overrightarrow{X}}$, the market index $S$ is defined as the weighted sum of the $d$ individual stock prices:
\[
S=\sum\limits_{i=1}^{d}w_iX_i,
\]
where weights $w_i$ can be interpreted as the total number of each stock available in the market.

Since a comonotonic random vector ${\overrightarrow{X^c}}:=(X_1^c, \cdots, X_d^c)$ can be represented as
\[
 {\overrightarrow{X^c}}\sim (F_1^{-1}(V), \cdots, F_d^{-1}(V)),
\]
the market index under the comonotonic stock prices assumption,
assuming the marginal distributions of individual stock prices to be unchanged, can be defined as
\[
S^c:=\sum\limits_{i=1}^{d}w_iF^{-1}_{X_i}(V)
\]

Noting the fact that, as shown in Remark \ref{order.var}, the variance of the market index
is maximized when the individual stock prices are comonotonic,
the herd behavior index by \citet{Dhaene2} is defined as
the ratio of variance of the market index to that of the index under the comonotonic assumption.
The following definition defines the simplified version of HIX. 
The original version of HIX defined using the option prices can be found in \citet{Dhaene2}.
\begin{equation*} 
\hix{\overrightarrow{w}, {\overrightarrow{X}}}:= \frac{{\rm Var}[S]}{{\rm Var}[S^c]}.
\end{equation*}
While HIX is a convenient measure which can measure the herd behavior effectively,
HIX may be sensitive to the marginal distributions \citep{Ahnhix}.
The revised version of  HIX (RHIX) defined as
\begin{equation*}
 \begin{aligned}
 \rhix{ \overrightarrow{w},{\overrightarrow{X}}}=\frac{\sum\limits_{i\neq j} w_iw_j{\rm cov}\left( X_i, X_j \right)}  {\sum\limits_{i\neq j} w_iw_j {\rm cov}\left( X_i^c, X_j^c \right)}.\\
\end{aligned}
\end{equation*}
is known to reduce the marginal distribution effects \citep{Ahnhix, Ahn10}. The same measure was proposed by \citet{Daniel14}
from a slightly different perspective.

Importantly, original definition of HIX (hence RHIX) can be calculated using the individual option prices and option price of the market
index \citep{Dhaene2, Daniel}, and these measures can be used as predictors of the degree of herd behaviors in the future as
implied by current option prices.
These are the main reasons why HIX and RHIX are favorable herd behavior indices, although
there may be some preference between HIX and RHIX. Of course, HIX and RHIX can also be estimated from the high frequency stock
market data \citep{Ahn10}.

\subsection{Marginal Dependency of RHIX}\label{sub.marginal}
Despite some controversy, if the perfect herd behavior corresponds to comonotonic movement \citep{Dhaene2},
the herd behavior should be a phenomenon that solely depends on the copula.
In this sense, RHIX may be more favorable than HIX, because it is known to reduce the marginal distribution effects \citep{Ahnhix}.
However, as expected from the definition of RHIX (it is defined based on covariances), RHIX cannot thoroughly remove the marginal effects.
Through a simple example, this section explains that such marginal effects in the calculation of RHIX can be arbitrarily large.

For expository purposes, we consider the following {\bf Toy Model} using the bivariate lognormal distribution, which
is frequently used to describe the stock prices.
\begin{toy*}
 Consider only two assets $\overrightarrow{X}=(X_1, X_2)$ that follow
a bivariate lognormal distribution with drift vector $\overrightarrow{u}$ and covolatility matrix ${\bf \Sigma}$
which is defined as $${\bf \Sigma}=\left(\begin{array}{c c}
                                                          \sigma_1^2 & \rho\sigma_1\sigma_2\\
                                                          \rho\sigma_1\sigma_2 & \sigma_2^2
                                                         \end{array}\right)
.$$
Note that $\overrightarrow{u}$, $\sigma_1$ and $\sigma_2$ are parameters related with marginal distributions, and
$\rho$ is the only parameter for the (Gaussian) copula; see, for example, \citet{Nelson} and \citet{Cherubini} for more details.
\end{toy*}
Under the {\bf Toy Model}, simple calculation shows that
\begin{equation}\label{rhix.2}
\rhix{{\bf w}, {\bf X}}=\frac{\exp(\rho_{12}\sigma_{1}\sigma_{2})-1}{\exp(\sigma_{1}\sigma_{2})-1}.
\end{equation}
We refer to \citet{Ahnhix} for more detailed calculation for HIX and RHIX under the lognormal model.
Now from \eqref{rhix.2}, the following equality shows that RHIX under the {\bf Toy Model }converges to $0$, when the common
volatility coefficient $\sigma_1=\sigma_2$ increases, regardless of
the copula coefficient $\rho$.
\begin{equation}\label{rhix.cov}
 \lim\limits_{\sigma_1=\sigma_2\rightarrow \infty}{\rm RHIX}\left(\overrightarrow{w}, \overrightarrow{X}\right)=0\quad \hbox{for any}
 \quad \rho<1.
\end{equation}
Knowing the degeneracy of RHIX as described in \eqref{rhix.cov}, the convergence rate can be important for the practical use of RHIX.
The following example shows the convergence rate of RHIX in \eqref{rhix.cov} under various circumstances.
\begin{example}\label{ex..1}
Figure \ref{figure.1}. (a) and (b) show
the variation of RHIX in the {\bf Toy Model} depending on the varying common volatility $\sigma:=\sigma_1=\sigma_2$
on the different scale time intervals of $(0,0.5)$ and $(0,5)$ respectively.
As shown in Figure \ref{figure.1}. (a), RHIX looks stable around reasonable weekly volatilities of the stock
markets assuming the weekly volatility to be $0.03$.\footnote{The weekly volatility for the S\&P 500 index and IBM from March to May of
2003 were $0.0309$ and $0.0365$ respectively. }
RHIX around yearly volatility ($=0.03\cdot\sqrt{52}\approx 0.22$)
even looks stable.
However, Figure \ref{figure.1}. (b) shows that RHIX slowly but surely decreases and converges to $0$ as $\sigma$ increases.
\end{example}

\begin{figure}[h!]
    \centering
    \subfloat[RHIX for various $\rho$ on the interval $(0,0.5)$]{%
      \includegraphics[width=0.39\textwidth]{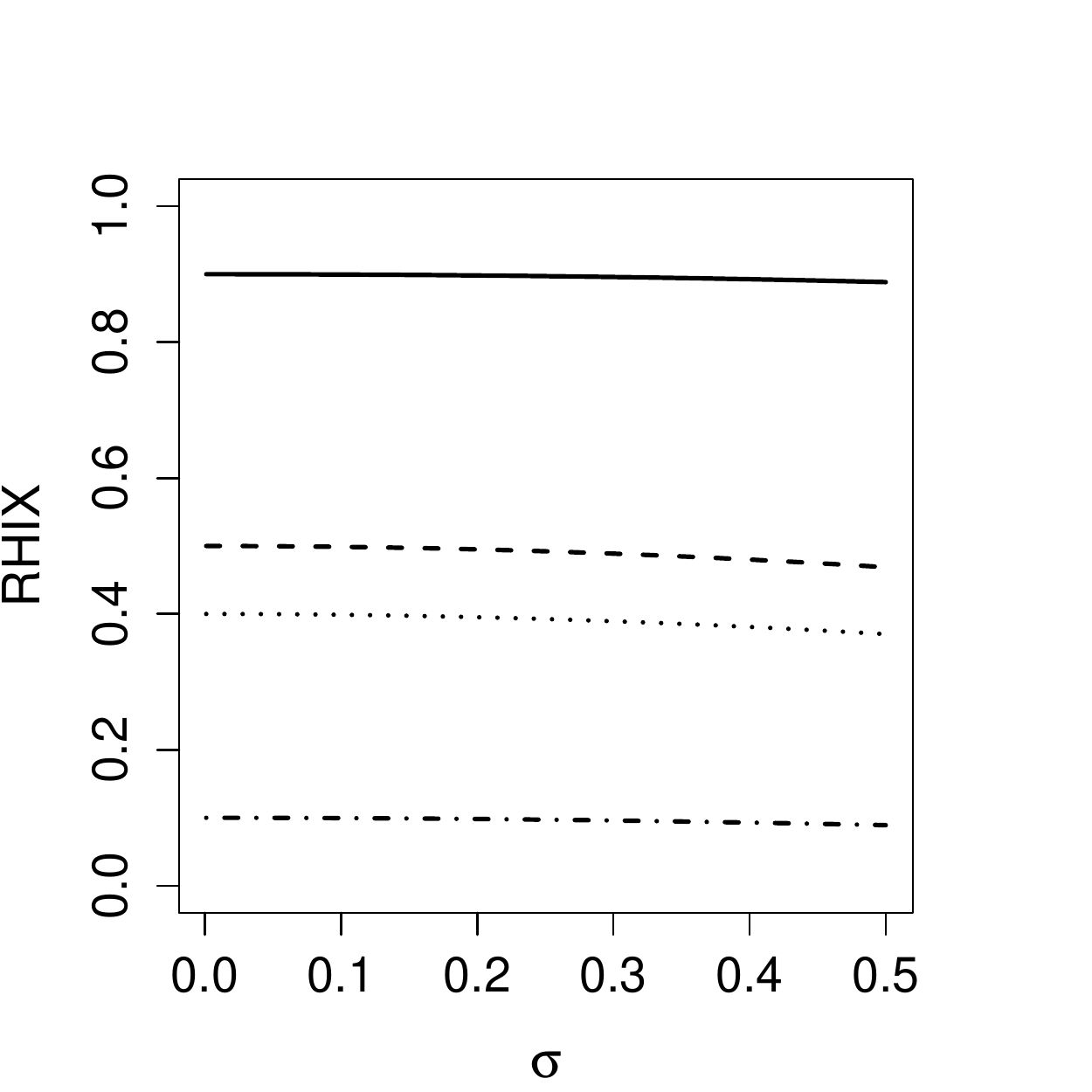}
    }
        \subfloat[RHIX for various $\rho$ on the interval $(0,5)$]{%
      \includegraphics[width=0.54\textwidth]{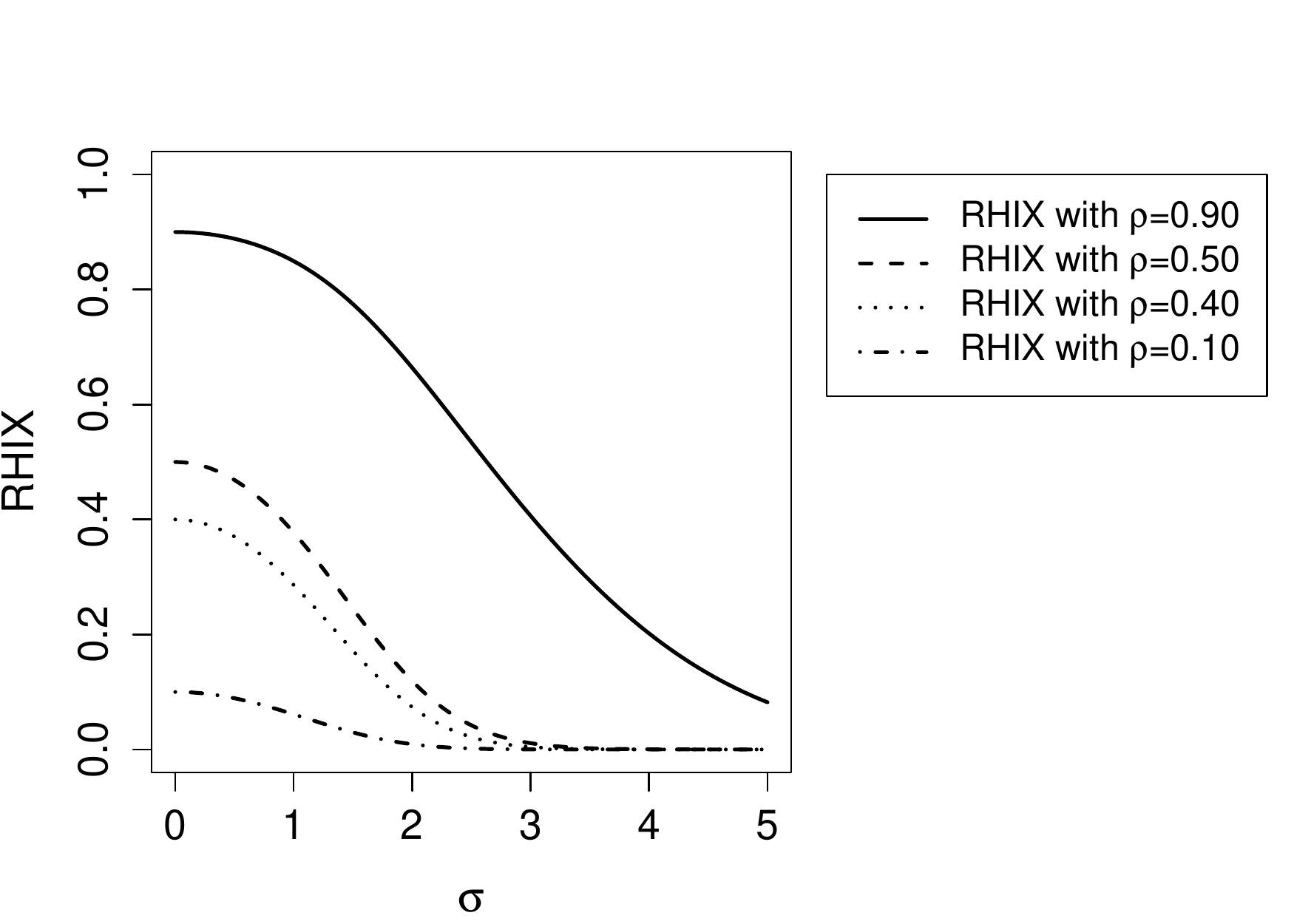}
    }
      \caption{RHIX with volatility effects.}
     \label{figure.1}
  \end{figure}

\subsection{New Herd Behavior Index: The Marginal Free Measure of Dependence}\label{def.six}
In the following definition, we propose a new herd behavior index that is free of marginal distribution and hence defined in terms of copula only.

\begin{definition}
For a given random vector $\overrightarrow{X}$, Spearman's rho type of the Herd Behavior Index (SIX) is defined as
\begin{equation*}
 \begin{aligned}
\six{\overrightarrow{w},{\overrightarrow{X}}}&:=\frac{\sum\limits_{i<j}  w_iw_j\rho_2\left( X_i, X_j \right)}  {\sum\limits_{i< j}  w_iw_j\rho_2\left( X_i^c, X_j^c \right)},\\
&=\frac{\sum\limits_{i< j} w_iw_j \rho_2\left( X_i, X_j \right)}  {\sum\limits_{i< j} w_iw_j},
 \end{aligned}
\end{equation*}
where Spearman's rho $\rho_2$ is defined as
\[
 \rho_2(X_i, X_j)=3\left(\P{(X_i-X_i^*)(X_j-X_j^{**})>0}- \P{(X_i-X_i^*)(X_j-X_j^{**})<0}\right)
\]
with $(X_i^*, X_j^*)$ and $(X_i^{**}, X_j^{**})$ are independent copies of $(X_i, X_j)$.
Sometimes we use $\six{\overrightarrow{w},H}$ to denote $\six{\overrightarrow{w},{\overrightarrow{X}}}$.
Note that SIX coincides with pairwise Spearman's rho defined in \citet{Schmid2}
with the equal weights $w_1=\cdots=w_d$.
\end{definition}

Since bivariate Spearman's rho does not depend on the marginal distribution, clearly SIX does not depend on the marginal distributions.
Hence, for continuous marginals, we have
\[
 \six{\overrightarrow{w},H}=\six{\overrightarrow{w},C}.
\]

Since SIX can be obtained by replacing the covariance terms in RHIX with the Spearman's rho terms,
it can be interpreted as the ratio of the weighted pairwise Spearman's rho of stock prices to
the weighted average of Spearman's rho of stock prices under the comonotonic assumption.
Furthermore, similar to RHIX as in \citet{Ahn10}, SIX can be expressed as the weighted average of the pairwise Spearman's rhos as shown below
\[
 \six{\overrightarrow{w},\overrightarrow{X}}=\E{Z}
\]
where
\[
 \P{Z=\rho_2(X_i, X_j)}=p_{i,j}
\]
with
\[
 p_{i,j}:=\frac{w_iw_j}{\sum\limits_{k\neq l}^{d} w_kw_l}.
\]

 Unlike HIX or RHIX, the calculation of SIX using the vanilla option prices may be difficult in reality
because, whereas the calculation of HIX and RHIX requires the option prices on the individual stocks and the market index, the calculation of SIX requires
the option prices related to every pairs of the individual stock prices.
As an alternative, high frequency stock price data can be used for the estimation of SIX: a detailed method for the estimation of
HIX and RHIX using high frequency stock price data can be found in \citet{Ahn10} and a similar method can be applied to the estimation of SIX.
Empirical analysis of herd behaviors in the stock market using stock price data and SIX can be found in \citet{Ahn7}.
\begin{remark}
In the calculation of HIX and RHIX, we have to calculate the variance or covariance under the comonotonic assumptions.
Hence, in the calculation of HIX and RHIX, an assumption on the marginal distributions
is essential as shown in \citet{Ahn10}, where lognormal distributions are assumed.
However, for the calculation of SIX, since Spearman's rho under comonotonic assumption is always $1$ regardless of the marginal distributions,
marginal assumption is not necessary.
\end{remark}

The following example present the representation of SIX in the multivariate log-normal distribution, and confirms
that SIX is free of marginal distribution.

\begin{example}\label{ex.formula}

For $\overrightarrow{w}\in(0,\infty)^d$ and $d$-variate log-normal random vector $\overrightarrow{X}=(X_1, \cdots,X_d)$ with drift vector $\overrightarrow{\mu}$ and covolatility matrix ${\mathbf \Sigma}$, SIX can be represented as
\begin{equation}\label{ex.formula.1}
\begin{aligned}
 \six{\bf w, {\bf X}} &=\frac{\sum\limits_{i\neq j}^{d} w_iw_j\frac{6}{\pi}{\rm acrsin}\left(\rho_{i,j} /2\right)     }{\sum\limits_{i \neq j}^{d} w_iw_j }  \\
  &=\sum\limits_{i\neq j}^{d} c_{i,j}\frac{6}{\pi}{\rm acrsin}\left(\rho_{i,j} /2\right)     \\
\end{aligned}
\end{equation}
where $ \rho_{i,j} =  \frac {{\mathbf \Sigma}_{i,j}}{\sqrt {{\mathbf \Sigma}_{i,i}{\mathbf \Sigma}_{j,j}}}$, and the first inequality is from
\citet{Kendall2}.
\end{example}

Spearman's rho preserves the concordance ordering, and one can easily expect that SIX also preserves the concordance ordering.
Hence it is possible to show that the maximum of SIX is achieved with the comonotonic copula.
However, due to the absence of the minimum copula in the Fr\'echet Space, the minimum of SIX is not as clear.
The following theorem provides some properties of SIX and determines the maximum and minimum of SIX.

\begin{theorem}\label{prop.six}
 For given $\overrightarrow{w}$, define $S_1:=w_1+\cdots+w_d$ and $S_2:=w_1^2+\cdots+ w_d^2$.
 Then, for the given distribution functions $H:=C(F_1, \cdots, F_d)$ and $H^*:=C^*(F_1, \cdots, F_d)$, the following holds:
 \begin{enumerate}
 \item[i.] If copulas $C,\,C^*\in\mathcal{F}_d$ satisfy $C \prec C^*$, then
 $\six{\overrightarrow{w}, H}\le \six{\overrightarrow{w}, H^*}$.
  \item[ii.] SIX satisfies
  \begin{equation}\label{six.eq.1}
\frac{1}{S_1^2-S_2}\left[ 12 \, l(\overrightarrow{w})-S_2\right] \le \six{\overrightarrow{w}, H} \le 1,
  \end{equation}
  where the definition of $l(\cdot)$ can be found in Notation \ref{ref.ref}.
\item[iii.] The upper bound of \eqref{six.eq.1} is attained if and only if\, $H$ is comonotonic.
\item[iv.] The lower bound of \eqref{six.eq.1} is attained if and only if
$H$ is $\overrightarrow{w^*}$-CM, where $\overrightarrow{w^*}$ is defined in \eqref{def.w} of Notation \ref{ref.ref}.

 \end{enumerate}
\end{theorem}

\begin{proof}
The proof of part i comes from the concordance property of Spearman's rho and the fact that SIX is a linear combination of bivariate Spearman's rho.
For the proof of the remaining parts,
note that
\[
\six{\overrightarrow{w}, H}=\six{\overrightarrow{w}, C}
\]
and
\begin{equation} \label {kang.eq1}
\begin{aligned}
\Var{\sum_{i=1}^{d} {w_i U_i}}
&=\E{\left( \sum_{i=1}^{d} {w_i U_i} -c \right)^2}\\
&=c^2 - 2 c\sum_{i=1}^d w_i\E{U_i} + \E{\left( \sum_{i=1}^{d} {w_i U_i} \right)^2}\\
&=c^2 - c\sum_{i=1}^d w_i + \sum_{i=1}^d w_i^2\E{U_i^2} + 2\sum_{i< j} w_i w_j \E{U_iU_j}  \\
&=\frac {1}{12}{ S_2} + 2\sum_{i< j} {w_i w_j {\rm Cov} (U_i,U_j)}\\
\end{aligned}
\end{equation}
where a constant $c$ is defined as $$c:=\frac{1}{2}S_1.$$
Now, Theorem \ref{lem.m1} and \eqref{kang.eq1} derive that
\begin{equation} \label {kang.eq100}
l(\overrightarrow{w}) \le\frac {1}{12}{ S_2} + 2\sum_{i< j} {w_i w_j {\rm Cov} (U_i,U_j)}
\le \frac{1}{12}\left( \sum\limits_{i=1}^{d}w_i\right)^2,
\end{equation}
where the first equality in the first inequality is achieved if and only if
$\overrightarrow{U}$ is $\overrightarrow{w^*}$-CM, and the second inequality is achieved
if and only if $\overrightarrow{U}$ is comonotonic.
Now \eqref{kang.eq100} and the following observation
\begin{equation*}
\begin{aligned}
 \rho_2 (X_i,X_j) &= \rho_2 (U_i,U_j)\\
 &=12\; Cov(U_i,U_j),
\end{aligned}
\end{equation*}
conclude the following inequalities
\begin{equation*}
   \frac{1}{S_1^2-S_2}\left[12\, l(\overrightarrow{w})-S_2\right] \le \six{\overrightarrow{w}, H} \le 1
   \end{equation*}
where the first equality holds if and only if
$H$ is $\overrightarrow{w^*}$-CM and the second inequality holds
if and only if $H$ is comonotonic.

\end{proof}

\subsection{Data Analysis}
In this subsection, we analyze the herd behaviors in the US stock market using SIX.
Daily stock prices $\overrightarrow{X}(t)$ of three stocks are collected from Apple, Hewlett-Packard Company and New York Times
in the time interval between $t=$2001/March/01 and $t=$2014/April/09.
Under the lognormal model, the line graph (${\rm SIX}$) in Figure \ref{figure.507} shows estimated SIX,
where SIX at each point is estimated based on 4 month observations.
Similar to \citet{Ahnhix}, three stock prices shows generally strong herd behavior during
the global financial crisis starting from 2008.

Sometimes, we may be interested in the relationship between the the stock prices of three companies only.
For example, we may assume that the stock prices of three companies reflect the preferences
between the traditional media system (newspapers), the traditional internet based media system (computers), and
the mobile internet based media system (smartphones or tablets).
However, strong dependency of the three stock prices may not stand for the strong dependency between three companies in particular,
because the strong comovement of the stock prices during the period
may be the result of the illusion effect caused by devaluation of the whole stock market (the global economic crisis in 2008, for example).
Hence, to understand the actual physical relation between three stock prices,
it can be beneficial to consider the detranded stock price by the market index
(S\&P in this data analysis) defined as follows:
$$\overrightarrow{X^M}(t):=\overrightarrow{X}(t)/S(t),$$
where $S(t)$ is S\&P index.
The dashed line graph (${\rm SIX}^M$)in Figure \ref{figure.507} shows estimated SIX using the adjusted stock price $\overrightarrow{X^M}$ on the same time interval.
Here, we have used the weight $\overrightarrow{w}=(1,1,1)$. Note that, under the lognormal model in \eqref{ex.formula},
specific statistical estimation procedures can be found in \citet{Ahn10}, for example.

From Figure \ref{figure.507}, we can conclude that main source of the comovement during the global financial crisis
is the devaluation of the whole stock market. After removing the comovement effect by the global financial crisis,  comovement of adjusted stock prices
$\overrightarrow{X}^M$ is not as strong.

   \begin{figure}[h!]
    \centering
      \includegraphics[width=0.7\textwidth]{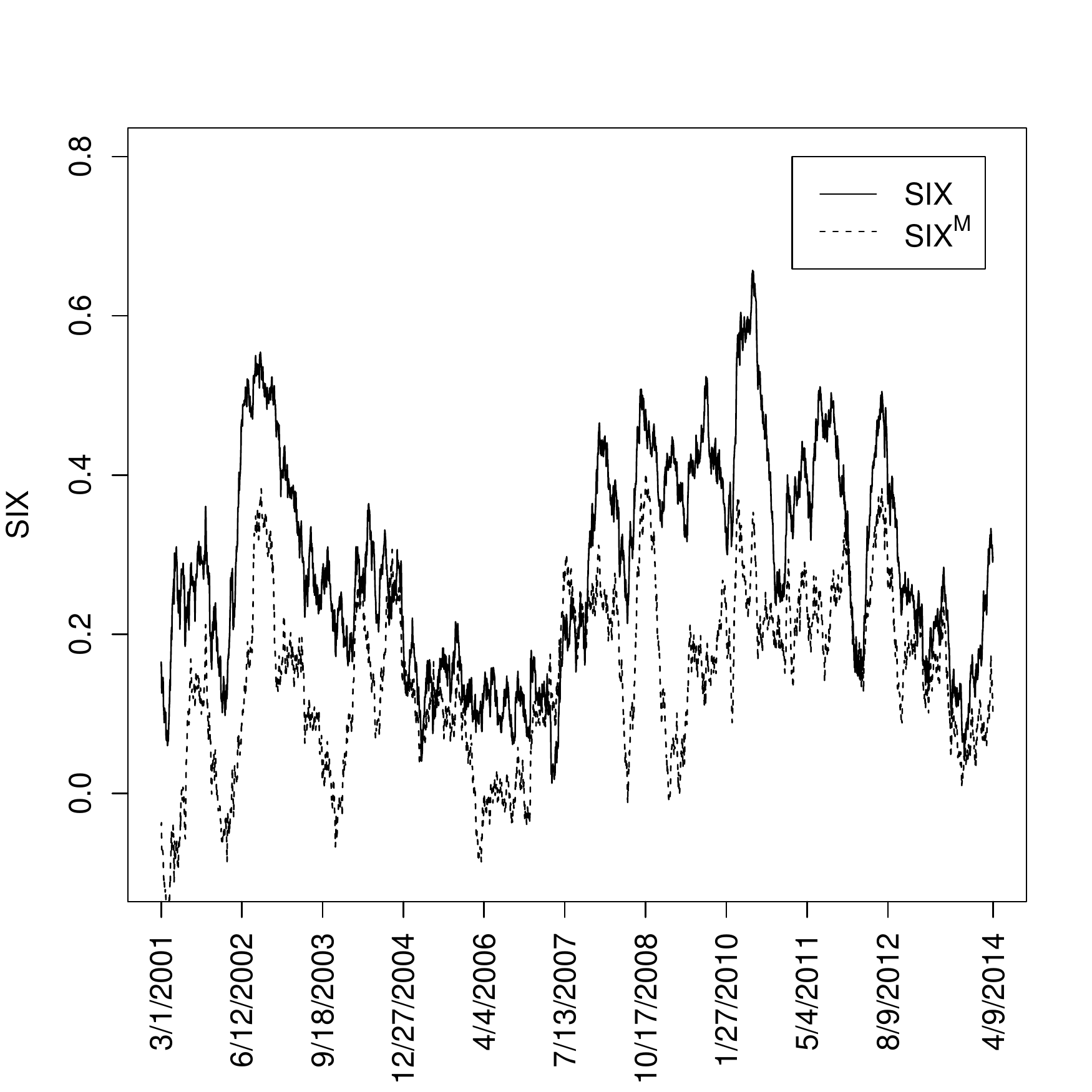}
      \caption{SIX from 2001/March/01 to 2014/April/09 with weight vector $(w_1, w_2, w_3)=(1,1,1)$.}
      \label{figure.507}
  \end{figure}

%
%
%
%
%

\section{Conclusion}
In this paper, we have provided the set of copulas called $\overrightarrow{w}$-CM copulas, and have shown these to be the minimal in set concordance ordering.
Given the absence of the minimum copula, the minimality can be important in optimization problems.
Especially, we show that the proposed set of copulas minimize the variance of the aggregated sum where
the marginal distributions are given as various uniform distributions.
As shown in Remark \ref{order.var} in Section \ref{optim}, the set of minimal copulas
can be related with the variance of aggregated sum with given marginal distributions.
In this respect, the approach using
$d$-CM copulas, which are the generalized version of $\overrightarrow{w}$-CM copulas, can be shown to be useful in minimizing
the variance of the aggregated sum for some special marginal distributions.
We leave this topic for future research.

Finally, although $\overrightarrow{w}$-CM copulas do not minimize the variance of the aggregated sum in general when the marginal
distributions are not uniform distributions, many other interesting optimization problems have uniform marginals as their solutions.
Optimization of the herd behavior index is one such example. In this paper, we have provided a herd behavior index that
does not depend on the marginal distributions, and showed that the herd behavior index is minimized with $\overrightarrow{w}$-CM copulas.

\section*{Acknowledgements}
For Jae Youn Ahn, this work was supported by the National Research Foundation of Korea(NRF) grant funded
by the Korean Government (2013R1A1A1076062).

\bibliographystyle{apalike}
\bibliography{CTE_Bib_HIX}

\begin{thebibliography}{}

\bibitem[Bernard et~al., 2014]{Ruodu5}
Bernard, C., Jiang, X., and Wang, R. (2014).
\newblock Risk aggregation with dependence uncertainty.
\newblock {\em Insurance Math. Econom.}, 54:93--108.

\bibitem[Cherubini et~al., 2004]{Cherubini}
Cherubini, U., Luciano, E., and Vecchiato, W. (2004).
\newblock {\em Copula methods in finance}.
\newblock Wiley Finance Series. John Wiley \& Sons Ltd., Chichester.

\bibitem[Cheung et~al., 2015]{Cheung6}
Cheung, K.~C., Denuit, M., and Dhaene, J. (2015).
\newblock Tail mutual exclusivity and tail-var lower bounds.
\newblock {\em FEB Research report AFI\_15100}.

\bibitem[Cheung et~al., 2011]{Cheung2}
Cheung, K.~C., Dhaene, J., and Tang, Q. (2011).
\newblock On partial hedging and counter-monotonic sums.
\newblock {\em Available at SSRN 1966995}.

\bibitem[Cheung and Lo, 2014]{Cheung4}
Cheung, K.~C. and Lo, A. (2014).
\newblock Characterizing mutual exclusivity as the strongest negative
  multivariate dependence structure.
\newblock {\em Insurance Math. Econom.}, 55:180--190.

\bibitem[Cheung and Vanduffel, 2013]{Cheung3}
Cheung, K.~C. and Vanduffel, S. (2013).
\newblock Bounds for sums of random variables when the marginal distributions
  and the variance of the sum are given.
\newblock {\em Scandinavian Actuarial Journal}, 2013(2):103--118.

\bibitem[Choi et~al., 2013]{Ahnhix}
Choi, Y., Kim, C., Lee, W., and Ahn, J.~Y. (2013).
\newblock Analyzing herd behavior in global stock markets: An intercontinental
  comparison.
\newblock {\em arXiv preprint arXiv:1308.3966}.

\bibitem[Cui and Sun, 2004]{Cui}
Cui, S. and Sun, Y. (2004).
\newblock Checking for the gamma frailty distribution under the marginal
  proportional hazards frailty model.
\newblock {\em Statistica Sinica}, 14(1):249--267.

\bibitem[Dhaene and Denuit, 1999]{Dhaene4}
Dhaene, J. and Denuit, M. (1999).
\newblock The safest dependence structure among risks.
\newblock {\em Insurance Math. Econom.}, 25(1):11--21.

\bibitem[Dhaene et~al., 2002a]{Dhaene5}
Dhaene, J., Denuit, M., Goovaerts, M.~J., Kaas, R., and Vyncke, D. (2002a).
\newblock The concept of comonotonicity in actuarial science and finance:
  applications.
\newblock {\em Insurance: Mathematics \& Economics}, 31(2):133--161.

\bibitem[Dhaene et~al., 2002b]{Dhaene}
Dhaene, J., Denuit, M., Goovaerts, M.~J., Kaas, R., and Vyncke, D. (2002b).
\newblock The concept of comonotonicity in actuarial science and finance:
  Theory.
\newblock {\em Insurance: Mathematics and Economics}, 31(1):3--33.

\bibitem[Dhaene et~al., 2012]{Dhaene2}
Dhaene, J., Linders, D., Schoutens, W., and Vyncke, D. (2012).
\newblock The {H}erd {B}ehavior {I}ndex: a new measure for the implied degree
  of co-movement in stock markets.
\newblock {\em Insurance: Mathematics and Economics}, 50(3):357--370.

\bibitem[Dhaene et~al., 2014a]{Daniel14}
Dhaene, J., Linders, D., Schoutens, W., and Vyncke, D. (2014a).
\newblock A multivariate dependence measure for aggregating risks.
\newblock {\em J. Comput. Appl. Math.}, 263:78--87.

\bibitem[Dhaene et~al., 2014b]{Dhaene3}
Dhaene, J., Linders, D., Schoutens, W., and Vyncke, D. (2014b).
\newblock A multivariate dependence measure for aggregating risks.
\newblock {\em J. Comput. Appl. Math.}, 263:78--87.

\bibitem[Dhaene et~al., 2006]{Goovaerts3}
Dhaene, J., Vanduffel, S., Goovaerts, M.~J., Kaas, R., Tang, Q., and Vyncke, D.
  (2006).
\newblock Risk measures and comonotonicity: a review.
\newblock {\em Stochastic Models}, 22(4):573--606.

\bibitem[Frees and Valdez, 1998]{Frees}
Frees, E.~W. and Valdez, E.~A. (1998).
\newblock Understanding relationships using copulas.
\newblock {\em North American Actuarial Journal}, 2(1):1--25.

\bibitem[Gaffke and R{\"u}schendorf, 1981]{Ruschendorf}
Gaffke, N. and R{\"u}schendorf, L. (1981).
\newblock On a class of extremal problems in statistics.
\newblock {\em Mathematische Operationsforschung und Statistik Series
  Optimization}, 12(1):123--135.

\bibitem[Genest et~al., 2007]{Genest}
Genest, C., Favre, A., B{\'e}liveau, J., and Jacques, C. (2007).
\newblock Metaelliptical copulas and their use in frequency analysis of
  multivariate hydrological data.
\newblock {\em Water Resources Research}, 43(9):W09401.

\bibitem[Joe, 1997]{Joe2}
Joe, H. (1997).
\newblock {\em Multivariate models and dependence concepts}, volume~73 of {\em
  Monographs on Statistics and Applied Probability}.
\newblock Chapman \& Hall, London.

\bibitem[Kaas et~al., 2002]{Kaas}
Kaas, R., Dhaene, J., Vyncke, D., Goovaerts, M.~J., and Denuit, M. (2002).
\newblock A simple geometric proof that comonotonic risks have the
  convex-largest sum.
\newblock {\em Astin Bull.}, 32(1):71--80.

\bibitem[Kendall and Gibbons, 1990]{Kendall2}
Kendall, M. and Gibbons, J.~D. (1990).
\newblock {\em Rank correlation methods}.
\newblock A Charles Griffin Title. Edward Arnold, London, fifth edition.

\bibitem[Kotz and Seeger, 1992]{Kotz2}
Kotz, S. and Seeger, J.~P. (1992).
\newblock Lower bounds on multivariate distributions with preassigned
  marginals.
\newblock In {\em Stochastic inequalities ({S}eattle, {WA}, 1991)}, volume~22
  of {\em IMS Lecture Notes Monogr. Ser.}, pages 211--218. Inst. Math.
  Statist., Hayward, CA.

\bibitem[Lee and Ahn, 2014a]{Ahn10}
Lee, W. and Ahn, J.~Y. (2014a).
\newblock Financial interpretation of herd behavior index and its statistical
  estimation.
\newblock {\em Journal of the Korean Statistical Society, In Press}.

\bibitem[Lee and Ahn, 2014b]{Ahn7}
Lee, W. and Ahn, J.~Y. (2014b).
\newblock On the multidimensional extension of countermonotonicity and its
  applications.
\newblock {\em Insurance: Mathematics and Economics}, 56:68--79.

\bibitem[Linders and Schoutens, 2014]{Daniel}
Linders, D. and Schoutens, W. (2014).
\newblock A framework for robust measurement of implied correlation.
\newblock {\em J. Comput. Appl. Math.}, 271:39--52.

\bibitem[Nelsen, 2006]{Nelson}
Nelsen, R.~B. (2006).
\newblock {\em An introduction to copulas}.
\newblock Springer Series in Statistics. Springer, New York, second edition.

\bibitem[Puccetti et~al., 2012]{Ruodu2}
Puccetti, G., Wang, B., and Wang, R. (2012).
\newblock Advances in complete mixability.
\newblock {\em Journal of Applied Probability}, 49(2):430--440.

\bibitem[Puccetti and Wang, 2014]{Ruodu3}
Puccetti, G. and Wang, R. (2014).
\newblock General extremal dependence concepts.
\newblock {\em Available at SSRN 2436392}.

\bibitem[R{\"u}schendorf and Uckelmann, 2002]{Ruschendorf2}
R{\"u}schendorf, L. and Uckelmann, L. (2002).
\newblock Variance minimization and random variables with constant sum.
\newblock In {\em Distributions with given marginals and statistical
  modelling}, pages 211--222. Kluwer Academic Publishers, Dordrecht.

\bibitem[Schmid and Schmidt, 2007]{Schmid2}
Schmid, F. and Schmidt, R. (2007).
\newblock Multivariate extensions of {S}pearman's rho and related statistics.
\newblock {\em Statistics \& Probability Letters}, 77(4):407--416.

\bibitem[Wang and Wang, 2011]{Ruodu}
Wang, B. and Wang, R. (2011).
\newblock The complete mixability and convex minimization problems with
  monotone marginal densities.
\newblock {\em Journal of Multivariate Analysis}, 102(10):1344--1360.

\bibitem[Wang and Wang, 2014]{Ruodu4}
Wang, B. and Wang, R. (2014).
\newblock Joint mixability.
\newblock {\em Preprint, University of Waterloo}.

\end{thebibliography}
\appendix
\section{}

\citet{Ahn7} proposed the class of minimal copulas
which can be viewed as alternatives to countermonotonicity in multivariate dimensions.

\begin{definition}[\citet{Ahn7}]\label{def.multi3}
A $d$-variate random vector $\overrightarrow{U}$ will be called $d$-countermonotonic
($d$-CM) or $d$-CM   if there exist function $(f_1, \cdots, f_d)\in\mathcal{M}_+^d[0,1]$
and
\begin{equation}\label{eq.def.5}
\sum\limits_{j=1}^d f_i(U_i)=c
\end{equation}
with probability $1$ for some constant $c\in\Real$.
Equivalently, we say that the distribution function $C$ is $d$-CM if $\overrightarrow{U}$ is $d$-CM.
Especially, for the choice of functions with $c=1$ in \eqref{eq.def.5},
$\overrightarrow{U}$ is called $d$-CM with parameter functions $(f_1, \cdots, f_d)$.
\end{definition}

Since \citet{Ahn7} have shown that $d$-CM does not depend on marginal distributions (see Lemma 1 in \citet{Ahn7}), we
provide a version of $d$-CM definition for a copula only in this appendix.
As we have briefly mentioned in Section \ref{sec.3},
$d$-CM may be too general to be used for the extreme negative dependence as it includes
almost countermonotonic movement.
Alternatively, \citet{Ahn7} provide a definition of strict $d$-CM as a subset of $d$-CM in the following sense.

\begin{definition}[\citet{Ahn7}]\label{def.multi4}
A $d$-variate random vector $\overrightarrow{U}$ is strict $d$-CM if
\begin{equation*}
\P{\sum\limits_{j=1}^d U_i=\frac{d}{2}}=1.
\end{equation*}
Equivalently, we say that $H$ is \it strict $d$-CM if $\overrightarrow{U}$ is strict $d$-CM.
\end{definition}
It is obvious that strictly $d$-CM is $d$-CM having constant multiplication of identity functions as parameter functions: i.e.
$$f_1(v)=\cdots=f_d(v)=\frac{2}{d}\cdot v$$
for $v\in[0,1]$.
The existence of a strict $d$-CM copula is shown in \citet{Ruschendorf2, Ahn7}.
Strict $d$-CM is useful in various minimization/maximization problems \citep{Ahn7}.

\begin{proof}[Proof of Corollary \ref{rem.2} ]
Showing Corollary \ref{rem.2} is equivalent to show that
for any given $\overrightarrow{w}$-CM copula $C$ and $C^*\in\mathcal{F}_d$ satisfying
$$C^*\prec C,$$
implies that $C^*$ is also $\overrightarrow{w}$-CM.

First observe that if $\overrightarrow{w}$ does not satisfy \eqref{cond.min}, then $\mathbb{C}$ is empty and the proof is trivial.
So we can assume that $\overrightarrow{w}$ satisfies \eqref{cond.min} and $\mathbb{C}$ is not empty.
 Now, define two sets
\[
 \mathcal{F}_c:=\left\{\overrightarrow{u}\in [0,1]^d\bigg\vert \sum\limits_{i=1}^{d}w_iu_i < \frac{\sum\limits_{i=1}^{d}w_i}{2}\right\},
\]
and
\[
 \mathcal{Q}_c:=\left\{\overrightarrow{u}\in [0,1]^d\bigg\vert \sum\limits_{i=1}^{d}w_iu_i < \frac{\sum\limits_{i=1}^{d}w_i}{2}\;\;\hbox{and}\;\; \hbox{$u_1, \cdots, u_d$ are rational numbers}\right\}.
\]
Then, by the denseness of rational numbers in real line, we have
 \[
 \begin{aligned}
   \left\{\overrightarrow{x}\in [0,1]^d \bigg\vert \sum\limits_{i=1}^{d}w_ix_i < \frac{\sum\limits_{i=1}^{d}w_i}{2} \right\} &=  \bigcup\limits_{\overrightarrow{u}\in  \mathcal{F}_c} \left\{\overrightarrow{x}\in [0,1]^d\bigg\vert  \overrightarrow{x} < \overrightarrow{u}\right\}\\
   &=  \bigcup\limits_{\overrightarrow{u}\in  \mathcal{Q}_c} \left\{\overrightarrow{x}\in [0,1]^d\bigg\vert  \overrightarrow{x} < \overrightarrow{u}\right\}.
 \end{aligned}
 \]
 which in turn implies
\begin{equation}\label{revision.3}
\begin{aligned}
 \P{  \sum\limits_{i=1}^{d}w_iU_i^* < \frac{\sum\limits_{i=1}^{d}w_i}{2}  } &= \P{\overrightarrow{U^*} \in  \left\{\overrightarrow{x}\in [0,1]^d \bigg\vert \sum\limits_{i=1}^{d}w_ix_i < \frac{\sum\limits_{i=1}^{d}w_i}{2} \right\}  }
 \\&=  \P{ \overrightarrow{U^*} \in \bigcup\limits_{\overrightarrow{u}\in  \mathcal{Q}_c} \left\{\overrightarrow{x}\in [0,1]^d\bigg\vert  \overrightarrow{x} < \overrightarrow{u}\right\}  }\\
 &\le  \sum\limits_{\overrightarrow{u}\in  \mathcal{Q}_c} \P{\overrightarrow{U^*}< \overrightarrow{u}}\\
 &=0.
\end{aligned}
\end{equation}
where the last inequality holds because $\mathcal{Q}_c$ is countable set.
Similar logic derives
\[
 \P{  \sum\limits_{i=1}^{d}w_iU_i^* > \frac{\sum\limits_{i=1}^{d}w_i}{2}  }=0
\]
which in turn concludes the proof with \eqref{revision.3}.
\end{proof}

\end{document}